\documentclass[reqno]{amsart}

\usepackage{amsmath,amsfonts,amsgen,amstext,amsbsy,amsopn,amsthm, amssymb}
\newtheorem{Lemma}{Lemma}
\newtheorem{Theorem}{THEOREM}
\newtheorem{proposition}{Proposition}
\newtheorem{assumption}{Assumption}

\theoremstyle{definition}

\theoremstyle{definition}

\theoremstyle{definition}

\newcommand\calC{{\mathcal{C}}}
\newcommand\id{\mathbb{I}}
\newcommand\nn\nonumber

\numberwithin{equation}{section}

\newcommand{\R}{\mathbb{R}}

\newcommand{\Z}{\mathbb{Z}}
\newcommand{\C}{\mathbb{C}}

\newcommand{\F}{\mathcal{F}}
\newcommand{\E}{\mathcal{E}}
\renewcommand{\H}{\mathcal{H}}

\newcommand{\be}{\begin{equation}}
\newcommand{\ee}{\end{equation}}
\newcommand{\bea}{\begin{align}}
\newcommand{\eea}{\end{align}}

\DeclareMathOperator{\tr}{tr}
\DeclareMathOperator{\Tr}{Tr}
\DeclareMathOperator{\const}{const}
\DeclareMathOperator{\im}{Im}
\DeclareMathOperator{\re}{Re}

\begin{document}

\title[Derivation of GL theory for a 1D system]{Derivation of Ginzburg-Landau 
theory\\ for a one-dimensional system with contact interaction}

\author[R.L. Frank]{Rupert L. Frank}
\address{{\rm (Rupert L. Frank)} Department of Mathematics, Princeton University, Princeton, NJ 08544, USA}
\email{rlfrank@math.princeton.edu}

\author[C. Hainzl]{Christian Hainzl}
\address{{\rm (Christian Hainzl)} Mathematisches Institut, Universit\"at T\"ubingen, Auf der Morgenstelle 10, 
72076 T\"ubingen, Germany}
\email{christian.hainzl@uni-tuebingen.de}

\author[R. Seiringer]{Robert Seiringer}
\address{{\rm (Robert Seiringer)} Department of Mathematics and Statistics, McGill University, 805 Sherbrooke Street West, Montreal, QC H3A 2K6, Canada}
\email{rseiring@math.mcgill.ca}

\author[J.P. Solovej]{Jan Philip Solovej}
\address{{\rm (Jan Philip Solovej)} Department of Mathematics, University of Copenhagen, Universitetsparken 5, DK-2100 Copenhagen, Denmark}
\email{solovej@math.ku.dk}

\begin{abstract}
In a recent paper \cite{FHSS} we give the first rigorous derivation of
the celebrated Ginzburg-Landau (GL) theory, starting from the microscopic  
Bardeen-Cooper-Schrieffer (BCS) model. 
Here we present our results in the simplified case of a one-dimensional
system of particles interacting via a $\delta$-potential.
\end{abstract}

\date {March 9, 2011} 
\maketitle

\section{Introduction and Main Results}

\subsection{Introduction}

In 1950 Ginzburg and Landau \cite{GL} presented the first satisfactory
mathematical description of the phenomenon of superconductivity.
Their model examined the {\em macroscopic} properties of a
superconductor in a phenomenological way, without explaining its
microscopic mechanism. In the GL theory the superconducting state is
represented by a complex order parameter $\psi(x)$, which is zero in
the normal state and non-zero in the superconducting state. The order
parameter $\psi(x)$ can be considered as a macroscopic wave-function
whose square $|\psi(x)|^2$ is proportional to the density of
superconducting particles.

In 1957 Bardeen, Cooper and Schrieffer \cite{BCS} formulated the first
{\em microscopic} explanation of superconductivity starting from a
first principle Hamiltonian. In a major breakthrough they realized
that this phenomenon can be described by the {\em
  pairing-mechanism}. The superconducting state forms due to an
instability of the normal state in the presence of an attraction
between the particles. In the case of a metal the attraction is made
possible by an interaction through the lattice. For other systems,
like superfluid cold gases, the interaction is of local type. In the
BCS theory the superconducting state, which is made up by pairs of
particles of opposite spin, the {\em Cooper-pairs}, is described by a
two-particle wave-function $\alpha(x,y)$.
 
A connection between the two approaches, the phenomenological GL
theory and the microscopic BCS theory, was made by Gorkov
\cite{gorkov} who showed that, close to the
critical temperature, the order parameter $\psi(x)$ and the
pair-wavefunction $\alpha(x,y)$ are proportional. 
A simpler argument was later given by de Gennes \cite{deGenne}.

Recently we presented in \cite{FHSS} a mathematical proof of the
equivalence of the two models, GL and BCS, in the limit when the
temperature $T$ is close to the critical temperature $T_c$, i.e., when
$ h = \left[(T_c -T)/T_c\right]^{1/2} \ll 1,$ where $T_c$ is the
critical temperature for the translation-invariant BCS equation. The
mathematical aspects of this equation where studied in detail in
\cite{HHSS,FHNS,HS,HS2,HS3}. In the present paper we present this
result in the simplified case a of one-dimensional system where the
particles interact via an attractive contact interaction potential of the form
\begin{equation}\label{deltav}
V(x-y) = - a \delta(x-y) \quad\text{with}\ a>0 \,.
\end{equation}

We assume that the system is subject to a weak external potential $W$,
which varies on  a large scale $1/h$ compared to the microscopic
scale of order $1$.  Since variations of the system on the macroscopic
scale cause a change in energy of the order $h^2$, we assume that the
external potential $W$ is also of the order $h^2$.  Hence we write it
as $h^2 W(hx)$, with $x$ being the microscopic variable. 
The parameter $h$ will play the role of a semiclassical parameter.

We will prove that, to leading order in $h$, the Cooper pair wave
function $\alpha(x,y)$ and the GL function $\psi(x)$ are related by
\begin{equation}\label{alphconverg}
\alpha(x,y) = \psi\left( h \frac{x+y}{2}\right) \alpha^0(x-y) 
\end{equation}
where $\alpha^0$ is the translation invariant minimizer of the BCS
functional.  In particular, the argument $\bar x$ of the order
parameter $\psi(\bar x)$ describes the {\em center-of-mass} motion of
the BCS state, which varies on the macroscopic scale. To be precise,
we shall prove that $ \alpha(x,y) = \tfrac 12 (\psi(hx) + \psi(hy))
\alpha^0(x-y) $ to leading order in $h$, which agrees with
\eqref{alphconverg} to this order.

For simplicity we restrict our attention to contact potentials of the
form (\ref{deltav}), but our method can be generalized to other kinds
of interactions; see \cite{FHSS} for details. The proof presented here
is simpler than the general proof in \cite{FHSS} which applies to any
dimension $d\leq 3$. There are several reasons for this.  First, there
is no magnetic field in one dimension. Second, for a contact
interaction the translation invariant problem is particularly simple
and the corresponding gap equation has an explicit solution. Finally,
several estimates are simpler in one dimension due the boundedness of
the Green's function for the Laplacian.


\subsection{The BCS Functional}

We consider a macroscopic sample of a fermionic system, in one spatial
dimension. Let $\mu\in\R$ denote the chemical potential and $T>0$ the
temperature of the sample. The fermions interact through the
attractive two-body potential given in (\ref{deltav}).  In addition,
they are subject to an external force, represented by a potential
$W(x)$.

In BCS theory the state of the system can be conveniently described in terms of a $2\times 2$ operator valued matrix
$$
\Gamma = \left( \begin{array}{cc} \gamma & \alpha \\ \bar\alpha & 1 -\bar\gamma \end{array}\right) 
$$
satisfying $0\leq \Gamma \leq 1$ as an operator on $L^2(\R)\oplus L^2(\R)$. The
bar denotes complex conjugation, i.e., $\bar \alpha$ has the integral kernel
$\overline {\alpha(x,y)}$. In particular, $\Gamma$ is assumed to be hermitian,
which implies that $\gamma$ is hermitian and $\alpha$ is symmetric (i.e,
$\gamma(x,y)=\overline{\gamma(y,x)}$ and $\alpha(x,y)=\alpha(y,x)$.) There are
no spin variables in $\Gamma$. The full, spin dependent Cooper pair wave
function is the product of $\alpha$ with an antisymmetric spin singlet.

We are interested in the effect of weak and slowly varying external
fields, described by a potential $h^2 W(hx)$.  In order to avoid
having to introduce boundary conditions, we assume that the system is
infinite and periodic with period $h^{-1}$. In particular, $W$ should
be periodic. We also assume that the state $\Gamma$ is periodic. The
aim then is to calculate the free energy per unit volume.

We find it convenient to do a rescaling and use macroscopic variables
instead of the microscopic ones. In macroscopic variables, the BCS
functional has the form
\begin{align}
\F^{\rm BCS}(\Gamma) := \Tr \left( - h^2 \nabla^2 -\mu + h^2 W(x)\right) \gamma - T\, S(\Gamma) - ah\int_{\calC}  |\alpha(x,x)|^2 \, dx \label{def:bcs}
\end{align}
where $\calC$ denotes the unit interval $[0,1]$. The entropy equals
$S(\Gamma) = - \Tr \Gamma \ln \Gamma$. The BCS state of the system is
a minimizer of this functional over all admissible $\Gamma$.

The symbol $\Tr$ in (\ref{def:bcs}) stands for the trace per unit
volume. More precisely, if $B$ is a periodic operator (meaning that it
commutes with translation by $1$), then $\Tr B$ equals, by definition,
the (usual) trace of $\chi B$, with $\chi$ the characteristic function
of $\calC$. The location of the interval is obviously of no
importance. It is not difficult to see that the trace per unit volume
has the usual properties like cyclicity, and standard inequalities
like H\"older's inequality hold. This is discussed in more detail in
\cite{FHSS}.

\begin{assumption}\label{ass:1}
We assume that $W$ is a bounded, periodic function with period $1$ and $\int_\calC
W(x) \,dx = 0$.
\end{assumption}


\subsubsection{The Translation-Invariant Case}\label{sec:translinv}

In the translation invariant case $W=0$ one can restrict $\F^{\rm BCS}$ to translation invariant states. We write a general translation invariant state in form of the $2 \times 2$ matrix 
\begin{equation}\label{gt}
\Gamma = \left( \begin{array}{cc} \tilde\gamma(-ih\nabla) & \tilde\alpha(-ih\nabla) \\ \overline{\tilde\alpha(-ih\nabla)} & 1 -\overline{\tilde\gamma(-ih\nabla)} \end{array}\right) \,,
\end{equation}
that is, $\gamma=[\Gamma]_{11}$ and $\alpha=[\Gamma]_{12}$ have integral kernels
$$
\gamma(x,y)= \frac{1}{2\pi} \int_{\R} \tilde\gamma(hp) e^{ip(x-y)} \,dp
\quad\text{and}\quad
\alpha(x,y)= \frac{1}{2\pi} \int_{\R} \tilde\alpha(hp) e^{ip(x-y)} \,dp \,.
$$
The fact that $\Gamma$ is admissible
means that $\tilde\alpha(p)=\overline{\tilde\alpha(-p)}$, that
$0\leq\tilde\gamma(p)\leq 1$ and $|\tilde\alpha(p)|^2\leq
\tilde\gamma(p)(1-\tilde\gamma(-p))$ for any $p\in\R$. For states of this form
the BCS functional becomes
\begin{equation}\label{eq:translinv}
\F^{\rm BCS}(\Gamma) = \int_\R (h^2 p^2 - \mu) \tilde \gamma(hp)\,\frac{dp}{2\pi} - T \int_\R S(\tilde \Gamma(hp)) \frac{dp}{2\pi} - ah \left| \int_\R \tilde \alpha(hp) \frac{dp}{2\pi} \right|^2\,,
\end{equation}
with $S(\tilde\Gamma(p)) = -\Tr_{\C^2} \tilde\Gamma(p) \ln
\tilde\Gamma(p)$ and $\tilde\Gamma(p)$ the $2\times 2$ matrix obtained
by replacing $-i\nabla$ by $p$ in (\ref{gt}).

In the following, we are going to summarize some well-known facts
about the translation invariant functional \eqref{eq:translinv}. For given $a>0$,  we
define the critical temperature $T_c>0$ by the equation
\begin{equation}\label{def:tc}
\frac 1a = \int_\R \frac {\tanh\big(\frac{p^2 - \mu}{2T_c}\big)}{p^2 - \mu} \frac{dp}{2\pi} \,.
\end{equation}
The fact that there is a unique solution to this equation follows from
the strict monotonicity of $t/\tanh t$ for $t>0$.   If $T\geq T_c$, then the minimizer of \eqref{eq:translinv}
satisfies $\tilde\alpha\equiv 0$ and
$\tilde\gamma(hp)=(1+\exp((h^2p^2-\mu)/T))^{-1}$. If $0<T<T_c$, on the other hand, then there is a unique solution $\Delta_0>0$ of the {\em BCS gap equation} 
\begin{equation}\label{BCSgapequ} 
\frac 1a 
= \int_\R \frac 1{K_T^0(p)} \frac{dp}{2\pi} \,,
\end{equation}
where 
\begin{equation}\label{eq:kt0}
K_T^0(p) = \frac{\sqrt{(p^2 - \mu)^2
+\Delta_0^2}}{\tanh\left(\tfrac1{2T}\,\sqrt{(p^2 - \mu)^2
+\Delta_0^2}\right)} \,.
\end{equation}
Moreover, the minimizer of \eqref{eq:translinv} is given by
\begin{equation}\label{eq:translinvmin}
\tilde\Gamma^0(hp) = \left( 1 + \exp\left(\tfrac1T H_{\Delta_0}^0(hp)\right)
\right)^{-1}
\end{equation}
with
$$
H_{\Delta_0}^0(p)= \left( \begin{array}{cc} p^2 - \mu & -\Delta_0 \\ -\Delta_0& - p^2 + \mu\end{array}\right) \,.
$$
Writing $\tilde\alpha^0(hp)=[\tilde\Gamma^0(hp)]_{12}$ one easily deduces from \eqref{eq:translinvmin} that
\begin{equation}\label{eq:alpha0expl}
\tilde\alpha^0(p) = \frac{\Delta_0}{2K_T^0(p)} \,.
\end{equation}

To summarize, in the case $W\equiv 0$ the functional $\F^{\rm BCS}$
has a minimizer $\Gamma^0$ for $0<T<T_c$ whose off-diagonal element
does not vanish and has the integral kernel
\begin{equation}\label{def:a0}
\alpha^0((x-y)/h) = \frac{\Delta_0}{2} \int_{\R} \frac{1}{K_T^0(hp)} e^{ip(x-y)}
\,\frac{dp}{2\pi} \,.
\end{equation}
We emphasize that the function $\alpha^0$ depends on $T$. For $T$
close to $T_c$, which is the case of interest, we have $\Delta_0\sim
\const (1-T/T_c)^{1/2}$.


\subsection{The GL Functional}

Let $\psi$ be a periodic function in $H_{\rm loc}^1(\R)$. For numbers
$b_1,b_3>0$ and $b_2\in\R$ the Ginzburg-Landau (GL) functional is given by
\begin{equation}\label{GLfunct}
\E(\psi) = \int_{\calC} \left(b_1 |\psi'(x)|^2 + b_2 W(x) |\psi(x)|^2 + b_3 \left(1 - |\psi(x)|^2\right)^2\right) dx \,.
\end{equation}
We denote its ground state energy by 
$$
E^{\rm GL} =  \inf \{\E(\psi)\,| \, \psi \in H^1_{\rm per} \,\} \,.
$$
Under our assumptions on $W$ it is not difficult to show that there is a
corresponding minimizer, which satisfies a second order differential equation
known as the GL equation. 


\subsection{Main Results}

Recall the definition of the BCS functional $\F^{\rm BCS}$ in (\ref{def:bcs}). 
We define the energy $F^{\rm BCS}(T,\mu)$ as the difference between the  infimum
of $\F^{\rm BCS}$ over all admissible $\Gamma$ and the free energy of the
normal state
\begin{equation}\label{gamma0}
\Gamma_0 := \left( \begin{array}{cc} \gamma_0 & 0 \\ 0 & 1 -\bar\gamma_0 \end{array}\right) 
\end{equation}
with $\gamma_0 = (1+e^{(-h^2\nabla^2 + h^2 W(x) - \mu)/T})^{-1}$. That is,
\begin{equation}\label{bcsen}
F^{\rm BCS}(T,\mu) = \inf_{\Gamma} \F^{\rm BCS}(\Gamma) - \F^{\rm BCS}(\Gamma_0) \,.
\end{equation}
Note that
\begin{equation}\label{f0}
\F^{\rm BCS}(\Gamma_0) = - T\, \Tr \ln\left( 1+ \exp\left(-\left( -h^2 \nabla^2 -\mu + h^2 W(x)  \right) \right) /T\right) \,.
\end{equation}
For small $h$ this behaves like an (explicit) constant times
$h^{-1}$. Under further regularity assumptions on $W$, (\ref{f0}) can
be expanded in powers of $h$. We do not need this, however, since we
are only interested in the difference $F^{\rm BCS}(T,\mu)$.

Since $\Gamma_0$ is an admissible state, one always has $F^{\rm BCS}(T,\mu)\leq
0$. If the strict inequality $F^{\rm BCS}(T,\mu) < 0$ holds, then the system is
said to be in a superconducting (or superfluid, depending on the physical
interpretation) state.

\begin{Theorem}\label{thm:main}
  Let Assumption \ref{ass:1} be satisfied, and let $T_c>0$ be the
  critical temperature in the translation invariant case, defined in
  (\ref{def:tc}). Let $D>0$. Then there are coefficients $b_1$, $b_2$
  and $b_3$, given explicitly in \eqref{eq:coeff1}--\eqref{eq:coeff3}
  below, such that
\begin{equation}\label{enthm}
  F^{\rm BCS}(T_c(1-Dh^2),\mu) =  h^{3} \left(E^{\rm GL} - b_3 \right)+ o(h^{3})
\end{equation}
as $h\to 0$. More precisely, the error term $o(h^3)$  satisfies
$$
-\const h^{3 + \tfrac 13} \leq o(h^{3}) \leq \const h^{5} \,.
$$

 Moreover, if $\Gamma$ is an approximate minimizer of $\F^{\rm BCS}$
  at $T=T_c(1-h^2 D)$, in the sense that
  \begin{equation}
    \F^{\rm BCS}(\Gamma)\leq \F^{\rm BCS}(\Gamma_0) + h^{3} \left( E^{\rm GL} - 
b_3 + \epsilon\right)
  \end{equation}
  for some small $\epsilon > 0$, then the corresponding $\alpha$ can
  be decomposed as
  \begin{equation}\label{thm:dec}
    \alpha(x,y)  = \frac 1 2 \big( \psi( x) + \psi(y) \big) \alpha^0( h^{-1}(x-y) )   + \sigma(x,y) 
  \end{equation}
  with $\E^{\rm GL}(\psi) \leq E^{\rm GL} + \epsilon + \const h^{\tfrac 13} $, $\alpha^0$ defined in (\ref{def:a0}), 
   and
  \begin{equation}\label{sigmab}
    \int_{\calC\times \R} |\sigma(x,y)|^2 \, dx\, dy \leq \const h^{3+\tfrac 13}\,.
\end{equation}

\end{Theorem}


\subsection{The coefficients in the GL functional}

In order to give explicit expressions for the coefficients in the GL functional we introduce the functions
$$
g_0(z) = \frac{\tanh(z/2)}{z} \,, \qquad
g_1(z) = \frac{ e^{2 z} - 2 z e^{z}-1}{z^2 (1+e^{z})^2}\,
,\qquad
g_2(z) = \frac{2 e^{z} \left( e^{ z}-1\right)}{z \left(e^{z}+1\right)^3}\,.
$$
Setting, as usual, $\beta_c=T_c^{-1}$ we define
$$
c= \frac { 2 \int_\R \left[g_0( \beta_c (q^2 - \mu) )  - \beta_c (q^2 - \mu) g_1 (\beta_c(q^2 - \mu)) \right]\,dq }
{\beta_c \int_\R\frac{  g_1( \beta_c(q^2 - \mu))}{q^2 - \mu}\, dq} \,.
$$
The three coefficients of the G-L functional turn out to be as follows,
\begin{equation}\label{eq:coeff1}
b_1 =  cD\ \frac {\beta_c^2}{16} \int_{\R} \left( g_1(\beta_c(q^2 -\mu)) + 2\beta_c q^2\, g_2(\beta_c (q^2-\mu)) \right) \, \frac{dq}{2\pi} \,,
\end{equation}
\begin{equation}\label{eq:coeff2}
b_2 =  cD\ \frac {\beta_c^2}{4} \int_{\R}  g_1(\beta_c (q^2-\mu))\, \frac{dq}{2\pi}
\end{equation}
and
\begin{equation}\label{eq:coeff3}
b_3 = (cD)^2\ \frac{\beta_c^2}{16} \int_{\R} \, \frac{g_1(\beta_c(q^2-\mu))}{q^2-\mu}\, \frac{dq}{2\pi} \,.
\end{equation}

We shall now discuss the signs of these coefficients. First note that
$g_0(z)-zg_1(z) = (z g_0(z))'>0$ and $g_1(z)/z>0$, which implies that $c>0$.
Using $g_1(z)/z>0$ again, we see that $b_3>0$. In contrast, the coefficient
$b_2$ may have either sign, depending on the value of $\beta_c\mu$ (which
depends on $a$ and $\mu$). The coefficient $b_1$ is again positive, as the
following computation shows: using the fact that $g_2(z)=g_1'(z)+(2/z)g_1(z)$ we
find
\begin{align*}
 b_1 & =  cD\ \frac {\beta_c^2}{16} \int_{\R} \left( g_1(\beta_c(q^2 -\mu)) +
2\beta_c q^2\, \left( g_1'(\beta_c (q^2-\mu))+\frac{2g_1(\beta_c
(q^2-\mu))}{\beta_c (q^2-\mu)} \right) \right) \, \frac{dq}{2\pi} \\
& = cD\ \frac {\beta_c^2}{16} \int_{\R} \left( g_1(\beta_c(q^2 -\mu)) +
q \frac{d}{dq} \left( g_1(\beta_c (q^2-\mu)) \right)       
+ 4 q^2 \, \frac{g_1(\beta_c (q^2-\mu))}{q^2-\mu} \right) \, \frac{dq}{2\pi} \\
& = cD\ \frac {\beta_c^2}{4} \int_{\R} q^2\,  \frac{g_1(\beta_c
(q^2-\mu))}{q^2-\mu} \, \frac{dq}{2\pi}\,.
\end{align*}
The claimed positivity is now again a consequence of $g_1(z)/z>0$.


\section{Sketch of the proof}

In the following we will consider temperatures $T=T_c (1- Dh^2)$. It is not difficult to see that the solution $\Delta_0$ of the BCS gap equation \eqref{BCSgapequ}
is of order $\Delta_0 = O(h).$

It is useful to rewrite the BCS functional in a more convenient way.
Define $\Delta$ to be the multiplication operator
$$ \Delta = \Delta(x) = -\psi(x) \Delta_0,$$
where $\Delta_0$ is the solution of the BCS equation \eqref{BCSgapequ} for temperature $T$,
and $\psi$ a periodic function in $H^2_{\rm loc}(\R)$. Define further 
\begin{equation}\label{hdelta}
H_\Delta = \left( \begin{array}{cc}  -h^2 \nabla^2  -\mu + h^2 W(x) & \Delta \\ \overline{\Delta} & h^2\nabla^2 +\mu - h^2 W(x) \end{array} \right).
\end{equation}
Formally,  we can write the BCS functional as 
\begin{multline} 
\F^{\rm BCS}(\Gamma)  =  - \Tr (- h^2 \nabla^2 - \mu + h^2W) +\frac 12 \Tr H_\Delta \Gamma  - T S(\Gamma)  \\
+ \frac {1}{4ha} \Delta_0^2 \int_\calC |\psi(x)|^2 dx- 
h a \int_{\calC} \left|\frac{\Delta_0 \psi(x) }{2ha}-  \alpha(x,x)\right|^2 dx\,.
\end{multline}
The first two terms on the right are infinite, of course, only their sum is well-defined.
For an upper bound, we can drop the very last term. The terms on the first line are minimized for   
$ \Gamma_\Delta = 1/(1 + e^{\frac 1 T H_\Delta})$, which we choose as a trial state. Then 
\begin{align}\label{uppboun}
F^{\rm BCS} (T,\mu) & \leq  \F^{\rm BCS}(\Gamma_\Delta) - \F^{\rm BCS}(\Gamma_0) \\
& \leq  - \frac T 2  \Tr\left[ \ln(1+e^{-\tfrac 1TH_\Delta})-\ln(1+e^{-\tfrac 1T H_0})\right] + \frac {1}{4ha} \Delta_0^2 \int_\calC |\psi(x)|^2 dx \,. \nonumber
\end{align}
To complete the upper bound, we have to evaluate $\Tr [
\ln(1+e^{-H_\Delta/T})-\ln(1+e^{- H_0/T})]$. This is
done via a contour integral representation and semiclassical types of
estimates.

The lower bound is divided into several steps.  We first aim at an a
priori bound on $\alpha$ for a general state $\Gamma$, which has lower
energy than the translation-invariant state. With $$H_{\Delta_0}^0 =
\left( \begin{array}{cc} -h^2 \nabla^2 - \mu & \Delta_0 \\ \Delta_0& -
    h^2 \nabla^2 + \mu\end{array}\right)\,,$$ we can rewrite the
BCS functional in the form
\begin{multline}\label{represequ2}
\F^{\rm BCS}(\Gamma)  =  - \Tr (- h^2 \nabla^2 - \mu + h^2W) +  
\frac 12 \tr H^0_{\Delta_0} \Gamma - T S(\Gamma)  \\ + h^2 \Tr W\gamma  + \frac {1}{4ha} \Delta_0^2 - 
h a \int_{\calC} \left|\frac{\Delta_0 }{2ha}-  \alpha(x,x)\right|^2 dx.
\end{multline}
{}From the BCS equation and the definition of $\alpha^0$ in  \eqref{def:a0} we conclude that 
\begin{equation}\label{al0}
\alpha^0(0)  = \frac 1{2\pi h} \int_\R \frac{\Delta_0}{2 K_T^0(p)} dp = \frac {\Delta_0}{2a h},
\end{equation}
and hence 
\begin{equation}\label{eq12}
\F^{\rm BCS} (\Gamma) - \F^{\rm BCS}(\Gamma^0) \geq \frac T2  \H(\Gamma,\Gamma^0) +h^2 \Tr W(\gamma - \gamma^0) - a h\int_\calC |\alpha(x,x)- \alpha^0(0)|^2 dx,
\end{equation}
where 
$ \H$ denotes the relative entropy 
\begin{align}\nonumber
\H(\Gamma,\Gamma^0) & =  \frac 2 T \left(\frac 12 \tr H_{\Delta_0}^0 \Gamma - T S(\Gamma) + \frac 12 \tr H_{\Delta_0}^0 \Gamma^0 - T S(\Gamma^0)\right) \\
& =  \Tr\left[ \Gamma \left( \ln\Gamma - \ln\Gamma^0\right) + 
(1-\Gamma) \left( \ln(1-\Gamma )- \ln (1-\Gamma^0)\right)\right].
\end{align}
Note that the left side of (\ref{eq12}) is necessarily non-positive
for a minimizing state $\Gamma$.

One of the essential steps in our proof, which is used on several
occasions, is Lemma \ref{lem:klein}. This Lemma presents a lower bound
on the relative entropy of the form
\begin{align}\nonumber
\H(\Gamma,\Gamma^0)& \geq \Tr\left[ H^0 \left( \Gamma - \Gamma^0\right)^2\right]  \\ & \quad + \frac 13 \frac {\left( \Tr \Gamma(1-\Gamma) - \Tr\Gamma_0(1-\Gamma^0)\right)^2}{\left| \Tr\Gamma(1-\Gamma) - \Tr\Gamma^0(1-\Gamma^0)\right| + \Tr \Gamma^0(1-\Gamma^0)},\label{essinequ}
\end{align}
where $H^0 = (1-2\Gamma_0)^{-1}\ln ( (1-\Gamma^0)/\Gamma^0) $. In our case here, it equals $K_T^0(-ih\nabla)/T$, with $K_T^0$ defined in (\ref{eq:kt0}). From \eqref{eq12} we deduce that for a minimizer $\Gamma$ 
\begin{multline}\label{bsinequ}
0 \geq \F^{\rm BCS} (\Gamma) - \F^{\rm BCS}(\Gamma^0) \geq  \Tr K^0_T(\gamma - \gamma^0)^2 +h^2 \Tr W(\gamma - \gamma^0) \\ + 
 \int_{\calC} \langle \alpha(\,\cdot\,,y) - \alpha^0(\tfrac{\,\cdot\, - y}h) | K^0_T (-ih\nabla)- a \delta (\tfrac{\,\cdot\, -y}h) | \alpha(\,\cdot\,,y)-\alpha^0(\tfrac{\,\cdot\, - y}h)\rangle \, dy \\ + 
   \frac 13 \frac{T \left( \Tr\left[ \gamma(1-\gamma)-\gamma^0(1-\gamma^0)  - |\alpha|^2 + |\alpha^0|^2\right]\right)^2}{ \left| \Tr\left[ \gamma(1-\gamma)-\gamma^0(1-\gamma^0)  -|\alpha|^2 +|\alpha^0|^2\right]  \right| + \Tr\left[\gamma^0(1-\gamma^0) - |\alpha^0|^2\right]},
 \end{multline}
where $\langle\cdot |\cdot\rangle$ denotes the inner product in $L^2(\R)$. 
Observe that the term in the second line is a convenient way to write $ \Tr K_T^0 (\alpha -\alpha^0)^*(\alpha - \alpha^0) - a h\int_\calC |\alpha(x,x)- \alpha^0(0)|^2 dx$.
{}From the first line on the right side and the Schwarz inequality together with the fact that $K^0_T  - a \delta \geq 0$, we obtain
first that $\Tr K^0_T(\gamma - \gamma^0)^2 \leq O(h^3)$. Together with the last line this further gives the a priori bound $\|\alpha\|_2^2 \leq O(h)$. 

Next, we use that $K^0_T - a \delta$ has $\alpha^0$ as unique zero
energy ground state, with a gap of order one above zero, and we can
further conclude from \eqref{bsinequ} that $\alpha$ is necessarily of
the form
$$\alpha(x,y) = \tfrac {1}{2} (\psi(x) + \psi(y))\alpha^0((x-y)/h) + \beta(x,y),$$
with $\|\beta\|_2^2 \leq O(h^3)$.  This information about the
decomposition of $\alpha$ then allows us to deduce, again by means of
a lower bound of the type \eqref{essinequ}, that the difference
$\F^{\rm BCS} (\Gamma) - \F^{\rm BCS}(\Gamma_\Delta)$ is very small
compared to $h^3$. This reduces the problem to the computation we
already did in the upper bound.


\section{Semiclassics}

One of the key ingredients in both the proof of the upper and the
lower bound are detailed semiclassical asymptotics for operators of
the form
\begin{equation}
 \label{eq:hdelta}
H_\Delta = \left( \begin{array}{cc}  -h^2 \nabla^2  -\mu + h^2 W(x) & \Delta(x)
\\ \overline{\Delta(x)} & h^2\nabla^2 +\mu - h^2 W(x) \end{array} \right) \,.
\end{equation}
Here $\Delta(x)=-h\psi(x)$ with a periodic function $\psi$, which is
of order one as $h\to 0$ (but might nevertheless depend on $h$). We
are interested in the regime $h\to 0$. In contrast to traditional
semiclassical results \cite{helffer, robert} we work under minimal smoothness assumptions on
$\psi$ and $W$. To be precise, we assume Assumption~\ref{ass:1} for
$W$ and that $\psi$ is a periodic function in $H_{\rm loc}^2(\R)$.

Our first result concerns the free energy.

\begin{Theorem}\label{sc}
Let 
\begin{equation}\label{deff}
f(z) = -   \ln \left(1+ e^{-z}\right) \ ,
\end{equation}
and define
\begin{equation}\label{defg0}
g_0(z) = \frac{f'(-z) - f'(z)}{z} = \frac { \tanh\left(\tfrac 12 z\right)}{z}\,,
\end{equation}
\begin{equation}\label{defg1}
g_1(z) = - g_0'(z) = \frac{f'(-z)-f'(z)}{z^2} + \frac{f''(-z)+f''(z)}{z} = \frac{ e^{2 z} - 2 z e^{z}-1}{z^2 (1+e^{z})^2}
\end{equation}
and 
\begin{equation}\label{defg2}
g_2(z) =  g_1'(z) + \frac 2 z \,g_1(z) =  \frac{f'''(z)-f'''(-z)}{z} = \frac{2 e^{z} \left( e^{ z}-1\right)}{z \left(e^{z}+1\right)^3}\,.
\end{equation}
Then, for any $\beta > 0$,
\begin{align}
\frac {h}{\beta}\, \Tr\left[ f(\beta H_\Delta) - f(\beta H_0)\right]  & = h^2
E_1 + h^4 E_2 + O(h^6)  \left( \|\psi\|_{H^1(\calC)}^6+
\|\psi\|_{H^2(\calC)}^2\right)\,, \label{210}
\end{align}
where
$$
E_1 = -\frac\beta 2 \|\psi\|_2^2 \int_\R g_0(\beta (q^2-\mu))\,\frac{dq}{2\pi}
$$
and
\begin{align*}
E_2 = & \frac{\beta^2}8 \|\psi'\|_2^2 \int_\R ( g_1(\beta (q^2-\mu))+2\beta q^2 g_2(\beta (q^2-\mu)) )\,\frac{dq}{2\pi} \\
& + \frac{\beta^2}2 \langle\psi|W|\psi\rangle \int_\R g_1(\beta (q^2-\mu)) \,\frac{dq}{2\pi} \\
& + \frac{\beta^2}8 \|\psi\|_4^4 \int_\R \frac{g_1(\beta (q^2-\mu))}{q^2-\mu} \,\frac{dq}{2\pi} \,.
\end{align*}
\end{Theorem}

More precisely, we claim that the diagonal entries of the $2\times2$
matrix-valued operator $f(\beta H_\Delta) - f(\beta H_0)$ are locally trace
class and that the sum of their traces per unit volume is given by
\eqref{210}. We sketch the proof of Theorem~\ref{sc} in Subsection \ref{sec:sc} below
and refer to \cite{FHSS} for some technicalities.

Our second semiclassical result concerns the behavior of $(1+
\exp(\beta H_\Delta))^{-1}$ in the limit $h\to 0$. More precisely, we
are interested in $[ (1+ \exp(\beta H_\Delta))^{-1}]_{12}$, where $[\,
\cdot\,]_{12}$ denotes the upper off-diagonal entry of an
operator-valued $2\times 2$ matrix. For this purpose, we define the
$H^1$ norm of a periodic operator $\eta$ by
\begin{equation}
 \label{eq:h1op}
\| \eta \|_{H^1}^2 = \Tr \left[ \eta^* (1-h^2\nabla^2)\eta \right] \,.
\end{equation}
In Subsection \ref{sec:scop} we shall prove

\begin{Theorem}\label{scop}
Let
\begin{equation}\label{defrho}
\rho(z) = \left(1+ e^{z}\right)^{-1}
\end{equation}
and let $g_0$ be as in \eqref{defg0}. Then
$$
\left[\rho(\beta H_\Delta)\right]_{12} = \frac{\beta h}4 \left(\psi(x)\, g_0(\beta(-h^2\nabla^2-\mu)) + g_0(\beta(-h^2\nabla^2-\mu))\, \psi(x) \right) + \eta_1 + \eta_2
$$
where
\begin{equation}\label{eq:eta1sc}
\| \eta_1 \|_{H^1}^2 \leq C h^5 \|\psi\|^2_{H^2(\calC)}
\end{equation}
and
\begin{equation}\label{eq:eta2}
\| \eta_2 \|_{H^1}^2 \leq C h^5 \left( \|\psi\|^2_{H^1(\calC)} + \|\psi\|^6_{H^1(\calC)} \right) \,.
\end{equation}
\end{Theorem}


\section{Upper Bound}

We assume that $T=T_c(1-Dh^2)$ with a fixed $D>0$ and denote by $\Delta_0$ the solution of the BCS gap equation \eqref{BCSgapequ}. In the following we write, as usual, $\beta=T^{-1}=\beta_c(1-Dh^2)^{-1}$ with $\beta_c=T_c^{-1}$. It is well known that the Ginzburg-Landau functional has a minimizer $\psi$, which is a periodic $H^2_{\rm loc}(\R)$ function. We put
$$
\Delta(x) = -\Delta_0 \psi(x) ,
$$
and define $H_\Delta$ by \eqref{hdelta}.

To obtain an upper bound for the energy we use the trial state
$$
\Gamma_\Delta = \left(1 + e^{\beta H_\Delta}\right)^{-1} \,.
$$
Denoting its off-diagonal element by $\alpha_\Delta=[\Gamma_\Delta]_{12}$, we have
the upper bound
\begin{align}\label{upb}
\F^{\rm BCS}(\Gamma_\Delta)  - \F^{\rm BCS}(\Gamma_0) 
& = -\frac 1{2\beta} \Tr\left[ \ln(1+e^{-\beta
H_\Delta})-\ln(1+e^{-\beta H_0})\right]  \nonumber \\
& \quad\, + \frac {\Delta_0^2}{4ha} \|\psi\|_2^2
-ha \int_C \left|\frac{\Delta_0\psi(x)}{2ha} -\alpha_\Delta(x,x)\right|^2 \,dx
\\ \nonumber & \leq -\frac 1{2\beta} \Tr\left[ \ln(1+e^{-\beta
H_\Delta})-\ln(1+e^{-\beta H_0})\right]  + \frac {\Delta_0^2}{4ha} \|\psi\|_2^2
\,.
\end{align}
The first term on the right side was evaluated in Theorem \ref{sc}. Applying this
theorem with $\psi$ replaced by $(\Delta_0/h)\psi$ we obtain that
\begin{align}\label{eq:upper1}
&\F^{\rm BCS}(\Gamma_\Delta)  - \F^{\rm BCS}(\Gamma_0)  \notag \\ 
& \leq - \frac{h\beta}4 \frac{\Delta_0^2}{h^2} \|\psi\|_2^2 \int_\R g_0(\beta
(q^2-\mu))\,\frac{dq}{2\pi} \notag \\
& \quad + \frac{h^3}{2} \left[
\frac{\beta^2}8 \frac{\Delta_0^2}{h^2} \|\psi'\|_2^2 \int_\R ( g_1(\beta (q^2-\mu))+2\beta q^2 g_2(\beta (q^2-\mu)) )\,\frac{dq}{2\pi} \right. \notag \\
& \quad\quad\quad + \left. \frac{\beta^2}2 \frac{\Delta_0^2}{h^2} \langle\psi|W|\psi\rangle \int_\R g_1(\beta (q^2-\mu)) \,\frac{dq}{2\pi} + \frac{\beta^2}8 \frac{\Delta_0^4}{h^4} \|\psi\|_4^4 \int_\R \frac{g_1(\beta (q^2-\mu))}{q^2-\mu} \,\frac{dq}{2\pi} \right] \notag \\
& \quad + \frac {\Delta_0^2}{4ha} \|\psi\|_2^2 + O(h^5) \,.
\end{align}
In the estimate of the remainder we used that $\psi$ is $H^2$ and that
$\Delta_0\leq Ch$.

Next, we use that by definition \eqref{BCSgapequ} of $\Delta_0$ the first and
the last term on the right side of \eqref{eq:upper1} cancel to leading order
and that one has
\begin{align*}
& - \frac{h\beta}4 \frac{\Delta_0^2}{h^2} \|\psi\|_2^2 \int_\R
g_0(\beta (q^2-\mu))\,\frac{dq}{2\pi} 
+ \frac {\Delta_0^2}{4ha} \|\psi\|_2^2 \\
& \quad = \frac{h\beta}4 \frac{\Delta_0^2}{h^2} \|\psi\|_2^2 \int_\R
\left(g_0(\beta \sqrt{(q^2-\mu)^2+\Delta_0^2}) - g_0(\beta (q^2-\mu))\right)
\,\frac{dq}{2\pi} \\
& \quad = - \frac{h^3 \beta^2}{8} \frac{\Delta_0^4}{h^4}  \|\psi\|_2^2 \int_\R
\frac{g_1(\beta (q^2-\mu))}{q^2-\mu} \,\frac{dq}{2\pi} + O(h^5) \,.
\end{align*}
We conclude that
\begin{align}\label{eq:upper2}
&\F^{\rm BCS}(\Gamma_\Delta)  - \F^{\rm BCS}(\Gamma_0)  \notag \\ 
& \leq \frac{h^3}{2} \left[
\frac{\beta^2}8 \frac{\Delta_0^2}{h^2} \|\psi'\|_2^2 \int_\R ( g_1(\beta (q^2-\mu))+2\beta q^2 g_2(\beta (q^2-\mu)) )\,\frac{dq}{2\pi} \right. \notag \\
& \quad\quad\quad + \frac{\beta^2}2 \frac{\Delta_0^2}{h^2} \langle\psi|W|\psi\rangle \int_\R g_1(\beta (q^2-\mu)) \,\frac{dq}{2\pi} + \frac{\beta^2}8 \frac{\Delta_0^4}{h^4} \|\psi\|_4^4 \int_\R \frac{g_1(\beta (q^2-\mu))}{q^2-\mu} \,\frac{dq}{2\pi} \notag \\
& \quad\quad\quad - \left. \frac{\beta^2}4 \frac{\Delta_0^4}{h^4} \|\psi\|_2^2 \int_\R \frac{g_1(\beta (q^2-\mu))}{q^2-\mu} \,\frac{dq}{2\pi} \right] + O(h^5) \,.
\end{align}
Up to an error of the order $O(h^5)$ we can replace $\beta=\beta_c(1-Dh^2)^{-1}$ by $\beta_c$ on the right side. Our last task is then to compute the asymptotics of $\Delta_0/h$. To do so, we rewrite the BCS gap equation \eqref{BCSgapequ} as
$$
\beta_c \int_\R g_0 (\beta_c (q^2-\mu)) \,\frac{dq}{2\pi}
= \frac 1a 
= \beta \int_\R g_0 (\beta \sqrt{(q^2-\mu)^2+\Delta_0^2}) \,\frac{dq}{2\pi} \,.
$$
A simple computation shows that
$$
\Delta_0^2 = D h^2 \ \frac { \int_\R \left[g_0( \beta_c (q^2 - \mu) )  - \beta_c (q^2 - \mu) g_1 (\beta_c(q^2 - \mu)) \right]\,dq }
{\beta_c \int_\R\frac{  g_1( \beta_c(q^2 - \mu))}{2 (q^2 - \mu)}\, dq} \left(1+ O(h^2)\right) \,.
$$
Inserting this into \eqref{eq:upper2} and using the fact that $\E(\psi) = E^{\rm GL}$ we arrive at the upper bound claimed in Theorem \ref{thm:main}.


\section{Lower Bound}

\subsection{The relative entropy}

As a preliminary to our proof of the lower bound, we present a general estimate for the relative entropy. In this subsection $H^0$ and $0\leq\Gamma\leq 1$ are arbitrary self-adjoint operators in a Hilbert space, not necessarily coming from BCS theory. Let $\Gamma^0:=\left(1 + \exp(\beta H^0)\right)^{-1}$. It is well-known that
$$
\H(\Gamma,\Gamma^0) = \Tr \left( \beta H^0 \Gamma + \Gamma\ln\Gamma+ (1-\Gamma)\ln(1-\Gamma) + \ln\left(1+ \exp(-\beta H_0)\right) \right)
$$
is non-negative and equals to zero if and only if $\Gamma=\Gamma^0$. Solving this equation for $H^0$, i.e., $H^0=\beta^{-1}( \ln(1-\Gamma^0) - \ln\Gamma^0)$, we can rewrite $\H(\Gamma,\Gamma^0)$ as a relative entropy,
\begin{align} \label{eq:relent}
\H(\Gamma,\Gamma^0) = \Tr\left[ \Gamma \left( \ln\Gamma - \ln\Gamma^0\right) + 
(1-\Gamma) \left( \ln(1-\Gamma )- \ln (1-\Gamma^0)\right)\right] \,.
\end{align}
The following lemma quantifies the positivity of $\H$ and improves an earlier result from \cite{HLS}.

\begin{Lemma}\label{lem:klein}
For any $0\leq \Gamma\leq 1$ and any $\Gamma_0$ of the form $\Gamma^0 = (1+e^{\beta H^0})^{-1}$, 
\begin{align*}
\H(\Gamma,\Gamma^0)& \geq  \Tr\left[ \frac {\beta H^0}{\tanh (\beta H^0/2)} \left( \Gamma - \Gamma^0\right)^2\right]  \\ & \quad + \frac 13 \frac {\left( \Tr \Gamma(1-\Gamma) - \Tr\Gamma^0(1-\Gamma^0)\right)^2}{\left| \Tr\Gamma(1-\Gamma) - \Tr\Gamma^0(1-\Gamma^0)\right| + \Tr \Gamma^0(1-\Gamma^0)} \,.
\end{align*}
\end{Lemma}

\begin{proof}
It is tedious, but elementary, to show that for real numbers $0<x,y<1$,
$$
x \ln\frac xy + (1-x) \ln \frac{1-x}{1-y} \geq \frac { \ln \frac{1-y}{y} } {1-2y} (x-y)^2 + \frac 13 \frac{ \left( x(1-x) - y(1-y) \right)^2}{\left| x(1-x) - y(1-y)\right| + y(1-y)}.
$$
Using joint convexity we see that
\begin{multline*}
 \frac{ \left( x(1-x) - y(1-y) \right)^2}{\left| x(1-x) - y(1-y)\right| + y(1-y)} \\= 4  \sup_{0<b<1} \left[ b(1-b) \left| x(1-x) - y(1-y)\right| - b^2 y(1-y) \right].
\end{multline*}
Let us replace on the right side the modulus $|a|$ by $\max\{a,-a\}$, and then use Klein's inequality \cite[Section 2.1.4]{thirring} for either of the expressions. This implies the result.
\end{proof}


\subsection{A priori estimates on $\alpha$}\label{sec:apriori}

We begin by briefly reviewing some facts about the translation-invariant case
$W\equiv 0$; see also Subsection \ref{sec:translinv}. Recall that $\Gamma^0$
denotes the minimizer of $\F$ in the translation-invariant case. It can be
written as $\Gamma^0=(1+e^{\beta H_{\Delta_0}^0})^{-1}$ with
$$
H_{\Delta_0}^0 = \left( \begin{array}{cc} -h^2 \nabla^2 -\mu 
& -\Delta_0 \\ -\Delta_0 &  h^2 \nabla ^2+\mu \end{array}\right) \,.
$$
Here $\Delta_0$ is the solution of the BCS gap-equation \eqref{BCSgapequ} and
$\beta^{-1}=T=T_c(1-Dh^2)$. Notice the distinction between $\Gamma^0$ and
$\Gamma_0$ which was defined in \eqref{gamma0}. The latter one, $\Gamma_0$,
contains the external potential $W$ and has no off-diagonal term.

Recall also that we denote the kernel of the off-diagonal entry
$\alpha^0=[\Gamma^0]_{12}$ by $\alpha^0((x-y)/h)$, which is explicitly
given in (\ref{def:a0}). From this explicit representation and the
fact that $\Delta_0 \leq C h$ we conclude, in particular, that
\begin{equation}\label{eq:alpha0}
\|\alpha^0\|_2^2 = \int_0^1 \,dy \int_\R \,dx \, |\alpha^0((x-y)/h)|^2 
= h \int_\R |\alpha^0(x)|^2 \,dx \leq C h \,.
\end{equation}
Moreover, the BCS gap-equation \eqref{BCSgapequ} is equivalent to
\begin{equation}\label{eq:bcsev}
(K_T^0(-ih\nabla) - ah \delta(x) )\alpha^0(x/h) = 0 \,.
\end{equation}
This implies, in particular, that 
\begin{equation}\label{eq:alpha0kin}
\Tr K^0_T(-ih\nabla) \alpha^0\overline{\alpha^0} - ah
|\alpha^0(0)|^2 =0 \,.
\end{equation}

Now we turn to the case of general $W$. Our goal in this subsection is to
prove that the $\alpha$ of \emph{any} low-energy state satisfies bounds similar
to \eqref{eq:alpha0} and \eqref{eq:alpha0kin}.

\begin{proposition}
Any admissible $\Gamma$ with $\F^{\rm BCS}(\Gamma) \leq \F^{\rm BCS}(\Gamma^0)$ satisfies
\begin{equation}
\label{eq:alpha}
\|\alpha\|_2^2 = \int_0^1 \,dx \int_\R \,dy \, |\alpha(x,y)|^2 \leq C h
\end{equation}
and
\begin{equation}
\label{eq:alphakin}
0\leq \Tr K^0_T(-ih\nabla) \alpha\overline{\alpha} - ah\int_{\calC} |\alpha(x,x)|^2\, dx \leq C h^{3} \,,
\end{equation}
where $\alpha=[\Gamma]_{12}$.
\end{proposition}

\begin{proof}
We divide the proof into two steps.

\emph{Step 1.}
Our starting point is the representation
\begin{align}\label{eq:repr}
\F^{\rm BCS}(\Gamma) - \F^{\rm BCS}(\Gamma^0) 
= \tfrac 1{2\beta} \, \H(\Gamma,\Gamma^0) + h^2 \Tr \gamma W - ah\int_{\calC}
|\alpha(x,x)- \alpha^0(0)|^2\, dx 
\end{align}
for any admissible $\Gamma$,  with the relative entropy $\H(\Gamma,\Gamma^0)$
defined in \eqref{eq:relent}. We note that $\Gamma^0$ is of the form $(1+e^{\beta
H^0})^{-1}$ with $H^0=H_{\Delta_0}^0$. We use Lemma~\ref{lem:klein} to bound
$\H(\Gamma,\Gamma^0)$ from below. Since $x\mapsto x/\tanh x$ is even, we can
replace $H^0_{\Delta_0}$ by its absolute value $E(-ih\nabla) = \sqrt{(-h^2\nabla^2 - \mu)^2
+ \Delta_0^2}$, and thus
\begin{align*}
& \Tr\left[ \frac {H^0_{\Delta_0}}{\tanh \tfrac \beta{2}  H^0_{\Delta_0}} \left(
\Gamma - \Gamma_0\right)^2\right] 
= \Tr \left[K^0_T(-ih\nabla)(\Gamma - \Gamma_0)^2 \right] \\
& \qquad \qquad = 2 \Tr K^0_T(-ih\nabla)(\gamma -\gamma^0)^2 
+ 2 \Tr K^0_T(-ih\nabla) (\alpha - \alpha^0) (\overline{\alpha - \alpha^0})
\end{align*}
with
$$
K^0_T(hp) = \frac {E(hp)}{\tanh \frac {\beta E(hp)}{2}}
$$
from \eqref{eq:kt0}. With the aid of  Lemma \ref{lem:klein} and the assumption
$\F^{\rm BCS}(\Gamma) \leq \F^{\rm BCS}(\Gamma^0)$ we obtain from
\eqref{eq:repr} the basic inequality 
\begin{align}\label{eq:basicineq}
0\geq 
& \Tr K^0_T(-ih\nabla)(\gamma -\gamma^0)^2 + h^2 \Tr \gamma W \notag \\
& + \Tr K^0_T(-ih\nabla) (\alpha - \alpha^0) (\overline{\alpha -
\alpha^0}) - ah\int_{\calC} |\alpha(x,x)- \alpha^0(0)|^2\, dx
\notag \\
&+ \frac T3 \frac{\left( \Tr\left[ \gamma(1-\gamma)-\gamma^0(1-\gamma^0) 
-\alpha\overline\alpha + \alpha^0 \overline{\alpha^0}\right]\right)^2}{ \left|
\Tr\left[ \gamma(1-\gamma)-\gamma^0(1-\gamma^0)  -\alpha\overline\alpha
+\alpha^0\overline{\alpha^0}\right]  \right| + \Tr\left[\gamma^0(1-\gamma^0) -
\alpha^0\overline{\alpha^0}\right]} \,.   
\end{align}
In the following step we shall derive the claimed a priori estimates on $\alpha$ from this inequality.

\emph{Step 2.}
We begin by discussing the first line on the right side of \eqref{eq:basicineq}. Using the fact that $\Tr W \gamma^0=0$ (since $W$ has mean value zero) and the Schwarz inequality we obtain the lower bound 
\begin{align}\label{eq:gammaap0}
& \Tr K^0_T(-ih\nabla)(\gamma -\gamma^0)^2 + h^2 \Tr W\gamma \notag \\
& \quad = \tfrac 12 \Tr K^0_T(-ih\nabla)(\gamma -\gamma^0)^2 + \tfrac 12 \Tr
K^0_T(\gamma -\gamma^0)^2 + h^2 \Tr W(\gamma - \gamma^0) \notag \\
& \quad \geq \tfrac 12 \Tr K^0_T(-ih\nabla)(\gamma -\gamma^0)^2 -  h^4\tfrac 12
\Tr W \left(K^0_T(-ih\nabla)\right)^{-1} W \notag \\
& \quad \geq \tfrac 12 \Tr K^0_T(-ih\nabla)(\gamma -\gamma^0)^2 -  C h^3 \,.
\end{align}
The last step used that $\Tr W \left(K^0_T(-ih\nabla)\right)^{-1} W\leq \|W\|_\infty^2 \Tr K^0_T(-ih\nabla)^{-1} \leq Ch^{-1}$.

Next, we treat the second line on the right side of \eqref{eq:basicineq}. Recall that the BCS gap equation in the form \eqref{eq:bcsev} says that the operator $K_T^0(-ih\nabla) - ah \delta(x)$ has an eigenvalue zero with eigenfunction $\alpha^0(x/h)$. Hence
\begin{align*}
& \Tr K^0_T(-ih\nabla) (\alpha - \alpha^0) (\overline{\alpha -
\alpha^0}) - ah\int_{\calC} |\alpha(x,x)- \alpha^0(0)|^2\, dx \\
&\quad = \Tr K^0_T(-ih\nabla) \alpha\overline\alpha - ah\int_{\calC}
|\alpha(x,x)|^2\, dx \,.
\end{align*}
Since a delta potential creates at most one bound state, zero must be the
ground state energy of $K_T^0(-ih\nabla) - ah \delta(x)$, and we deduce that
$$
\Tr K^0_T(-ih\nabla) \alpha\overline\alpha - ah\int_{\calC} |\alpha(x,x)|^2\, dx
\geq 0 \,.
$$
This information, together with \eqref{eq:gammaap0} and \eqref{eq:basicineq}, yields
\begin{equation}\label{eq:gammaap}
\Tr (1-h^2\nabla^2)(\gamma -\gamma^0)^2 \leq C \Tr K^0_T(-ih\nabla)(\gamma -\gamma^0)^2 \leq C h^{3} \,,
\end{equation}
\begin{equation}\label{eq:alphakin0}
\Tr K^0_T(-ih\nabla) \alpha\overline\alpha - ah\int_{\calC} |\alpha(x,x)|^2\, dx
\leq C h^{3}
\end{equation}
and
\begin{equation}\label{eq:frac}
\frac{\left( \Tr\left[ \gamma(1-\gamma)-\gamma^0(1-\gamma^0) 
-\alpha\overline\alpha + \alpha^0\overline{\alpha^0}\right]\right)^2}{ \left|
\Tr\left[ \gamma(1-\gamma)-\gamma^0(1-\gamma^0)  -\alpha\overline\alpha
+\alpha^0\overline{\alpha^0}\right]  \right| + \Tr\left[\gamma^0(1-\gamma^0) -
\alpha^0\overline{\alpha^0}\right]} \leq C h^3 \,.
\end{equation}
We know that $\Tr\left[\gamma^0(1-\gamma^0) - \alpha^0\overline{\alpha^0}
\right]\leq C h^{-1}$ from the explicit solution in the translation invariant
case, and therefore \eqref{eq:frac} yields
\begin{equation}\label{eq:frac1}
\left| \Tr\left[ \gamma(1-\gamma)-\gamma^0(1-\gamma^0)  -\alpha
\overline{\alpha} + \alpha^0\overline{\alpha^0} \right]  \right|
\leq Ch \,.
\end{equation}
In order to derive from this an a priori estimate on $\alpha$ we use \eqref{eq:gammaap} and the Schwarz inequality to bound
$$
\left| \Tr(\gamma-\gamma^0)\right| \leq h^{-2} \Tr
K_T^0(-ih\nabla)(\gamma-\gamma^0)^2 + h^2 \Tr
\left(K_T^0(-ih\nabla)\right)^{-1} \leq C h
$$
and
$$
\left| \Tr(\gamma^2 - (\gamma^0)^2) \right| =
\left|\Tr(\gamma-\gamma^0)(\gamma+\gamma^0)\right| \leq h^{-2}
\Tr(\gamma-\gamma^0)^2 + h^2\Tr (\gamma+\gamma^0)^2 \leq C h \,.
$$
Finally, since $\Tr\alpha^0\overline{\alpha^0} \leq C h$ (see
\eqref{eq:alpha0}), we conclude from \eqref{eq:frac1} that
$\Tr\alpha\overline{\alpha} \leq Ch$, as claimed.
\end{proof}


\subsection{Decomposition of $\alpha$}

Here we quantify in which sense $\alpha(x,y)$ is close to $\tfrac 12 \left(
\psi( x)+\psi(y)\right) \alpha^0(h^{-1}(x-y))$. There is one technical point
that we would like to discuss before stating the result. The asymptotic form
$\tfrac 12 \left( \psi( x)+\psi(y)\right) \alpha^0(h^{-1}(x-y))$ will allow us
in the next subsection to use the semiclassical results in a similar way as in the
proof of the upper bound. Our semiclassics, however, require $\psi$ to be in
$H^2$. While we naturally get an $H^1$ condition, the $H^2$ condition is
achieved by introducing an additional parameter $\epsilon>0$, which will later
chosen to go to zero as $h\to 0$.

\begin{proposition}\label{decomp}
Let $\Gamma$ be admissible with $\F^{\rm BCS}(\Gamma) \leq \F^{\rm BCS}(\Gamma^0)$. Then for every sufficiently small $\epsilon\geq h>0$, the operator $\alpha=[\Gamma]_{12}$ can be decomposed as
\begin{equation}\label{eq:alphadecomp}
\alpha(x,y) = \tfrac 12 \left( \psi( x)+\psi(y)\right)  \alpha^0(h^{-1}(x-y))  +
\sigma(x,y)
\end{equation}
with a periodic function $\psi\in H^2(\mathcal C)$ satisfying
\begin{equation}\label{eq:psi}
\|\psi\|_{H^1} \leq C\,,
\qquad
\|\psi\|_{H^2} \leq C\epsilon h^{-1}
\end{equation}
and with
\begin{equation}\label{eq:beta}
\|\sigma\|_{H^1}^2 \leq C\epsilon^{-2} h^3 \,.
\end{equation}
More precisely, one has $\sigma = \sigma_1 + \sigma_2$ with
\begin{equation}\label{eq:beta1}
\|\sigma_1\|_{H^1}^2 \leq C h^3
\end{equation}
and with $\sigma_2$ of the form
$$
\sigma_2(x,y) = \tfrac 12 \left(\tilde\psi( x)+\tilde\psi(y)\right)
\alpha^0(h^{-1}(x-y))\,,
$$
where the Fourier transform of $\tilde\psi$ supported in $\{|p|\geq\epsilon
h^{-1}\}$. The Fourier transform of $\psi$ is supported in $\{|p|<\epsilon
h^{-1}\}$.
\end{proposition}

We recall that the $H^1$ norm of an operator was introduced in \eqref{eq:h1op}.

\begin{proof}

\emph{Step 1.}
We can write \eqref{eq:alphakin} as
\begin{equation}\label{eq:alphakin2}
0\leq \int_{\calC} \langle \alpha(\,\cdot\,,y)  | K^0_T (-ih \nabla)- ah \delta (\cdot\,-y) | \alpha(\,\cdot\,,y)\rangle \, dy \leq C h^3 \,.
\end{equation}
Here, the operator $K^0_T(-ih\nabla)$ acts on the $x$ variable of $\alpha(x,y)$, and $\langle\, \cdot \, | \, \cdot \,\rangle$ denotes the standard inner product on $L^2(\R)$.

Now we recall that the operator $ K^0_{T} - ah \delta(\cdot\,-y)$ on $L^2(\R)$ has a unique ground state, proportional to $\alpha^0(h^{-1}(\cdot\,-y))$, with ground state energy zero. There are no further eigenvalues and the bottom of its essential spectrum is $\Delta_0/\tanh\left[\frac{\Delta_0}{2T}\right]\geq 2T$. In particular, there is a lower bound, independent of $h$, on the gap. We write 
\begin{equation}\label{eq:psi0def}
\psi_0(y) =  \left( h \int_{\R^d} |\alpha^0(x)|^2 \, dx  \right)^{-1} \int_{\R^d} \alpha^0(h^{-1}(x-y)) \alpha(x,y) \,dx
\end{equation}
and decompose
$$\alpha(x,y)=\psi_0(y) \alpha^0(h^{-1}(x-y))  + \sigma_0(x,y) \,.
$$
Then \eqref{eq:alphakin2} together with the uniform lower bound on the gap of
$K^0_{T} - ah\delta (\,\cdot\,-y)$ yields the bound $\|\sigma_0\|_2^2 \leq C
h^3$. We can also symmetrize and write
\begin{equation}\label{eq:beta1def}
\sigma_1(x,y) = \sigma_0(x,y) + \tfrac 12 \left(\psi(x) - \psi(y)\right)
\alpha^0(h^{-1}(x-y)) \,.
\end{equation}
Then
\begin{equation}\label{betas}
\alpha(x,y) = \tfrac 12 \left( \psi_0( x)+\psi_0(y)\right) \alpha^0(h^{-1}(x-y)) 
+ \sigma_1(x,y)
\end{equation}
again with
\begin{equation}\label{eq:beta1proof}
\|\sigma_1\|_2^2\leq C h^3 \,.
\end{equation}
This proves the first half of \eqref{eq:beta1}. Before proving the second half in Step 4 below we need to study $\psi$.

\emph{Step 2.} We claim that
\begin{equation}\label{eq:psi0}
\int_\calC |\psi_0(x)|^2 \,dx \leq C
\end{equation}
and
\begin{equation}\label{eq:psi0kin}
\int_\calC |\psi_0'(x)|^2 \,dx \leq C \,.
\end{equation}
The first inequality follows by Schwarz's inequality
$$
\int_\calC | \psi_0(x)|^2\, dx \leq \frac{\|\alpha\|_2^2}{h \int_{\R} |\alpha^0(x)|^2 dx }
$$
and our bounds \eqref{eq:alpha} and \eqref{eq:alpha0}. In order to prove \eqref{eq:psi0kin} we use again Schwarz's inequality, 
\begin{equation}\label{schw2}
\int_\calC |\psi_0'(x)|^2\, dx \leq  \frac {\int_{\R\times \calC}  \left|
\left(\nabla_x + \nabla_y\right) \alpha(x,y)\right|^2 \, dx\,dy }{h \int
|\alpha^0(x)|^2dx } \,.
\end{equation}
Lemma \ref{lem:P} below bounds the numerator by a constant times
$$
h^{-2} \int_{\calC} \langle \alpha(\,\cdot\,,y) | K_{T}^{0}(-ih\nabla)  - ah \delta(\cdot\,-y) | \alpha(\,\cdot\,,y)\rangle \, dy \,,
$$
and therefore \eqref{eq:psi0kin} is a consequence of \eqref{eq:alphakin2} and \eqref{eq:alpha0}.

\emph{Step 3.}
Next, we establish the remaining bound $\|\nabla\sigma_1\|^2_{2}\leq Ch$ in
\eqref{eq:beta1}. We use formula \eqref{eq:beta1def} for $\sigma_1$. First of
all, using the fact that $K(-ih\nabla)\geq c(1-h^2\nabla^2)$ one easily deduces
from \eqref{eq:alphakin2} that $\|\nabla \sigma_0\|_2^2\leq C h$. Moreover,
because of \eqref{eq:alpha0} and \eqref{eq:psi0kin}
$$
\int_\calC |\psi_0'(x)|^2 \int_{\R} |\alpha^0(h^{-1}(x-y))|^2 dx\, dy \leq C h \,.
$$
Finally,
\begin{align*}
&h^{-2} \int_{\calC\times\R} |\psi_0(x)-\psi_0(y)|^2 |(\alpha^0)'(h^{-1}(x-y))|^2\,dx \, dy \\ 
& = h^{-1} 4 \sum_{p\in 2\pi\Z} |\hat \psi_0(p)|^2 \int_{\R} \left|(\alpha^0)'(x)\right|^2 \sin^2\left( \tfrac 12 h p x \right) \, dx \\
& \leq h \sum_{p\in 2\pi\Z}|p|^2 |\hat \psi_0(p)|^2 \int_{\R} \left|(\alpha^0)'(x)\right|^2 x^2 \, dx 
\leq C h\,,
\end{align*}
where we used \eqref{eq:psi0kin} and the fact that $\int_\R \left|(\alpha^0)'(x)\right|^2 x^2 dx$ is finite. This is a simple consequence of the fact that the Fourier transform of $\alpha^0$
is given by the smooth function $\frac{\Delta_0}{2(2\pi)^{1/2}} \left(K_T^0\right)^{-1}$. 
This completes the proof of \eqref{eq:beta1}.

\emph{Step 4.}
Finally, for each $\epsilon\geq h$ we decompose $\psi_0 = \psi+\tilde\psi$,
where the Fourier transforms of $\psi$ and $\tilde\psi$ are supported in $\{|p|<
\epsilon h^{-1} \}$ and $\{|p|\geq \epsilon h^{-1} \}$, respectively. Clearly,
the bounds \eqref{eq:psi0} and \eqref{eq:psi0kin} imply \eqref{eq:psi}.

Moreover, $\|\tilde\psi\|_2 \leq C \epsilon^{-1} h$ and $\|\tilde\psi'\|_2
\leq C$, and hence
$$
\sigma_2(x,y)= \tfrac 12 \left(\tilde\psi( x)+\tilde\psi(y)\right)
\alpha^0(h^{-1}(x-y))
$$
satisfies $\|\sigma_2\|_{H^1}^2 \leq h^{3} \epsilon^{-2}$. This completes the
proof
of the proposition.
\end{proof}

In the previous proof we made use of the following

\begin{Lemma}\label{lem:P} For some constant $C>0$,
\begin{align}\nonumber
& h^2 \int_{\R\times \calC}  \left| \left(\nabla_x + \nabla_y\right) \alpha(x,y)\right|^2 \, dx\,dy 
\\ & \leq C \int_{\calC} \langle \alpha(\,\cdot\,,y) | K_{T}^{0}(-ih\nabla)  - ah \delta(\cdot\,-y) | \alpha(\,\cdot\,,y)\rangle \, dy \label{lem:eq:com}
\end{align}
for all periodic and symmetric $\alpha$ (i.e., $\alpha(x,y)=\alpha(y,x)$).
\end{Lemma}

\begin{proof}
By expanding $\alpha(x,y)$ in a Fourier series
\begin{equation}
\alpha(x,y) = \sum_{ p \in 2\pi \Z} e^{i p (x+y)/2} \alpha_p(x-y)
\end{equation}
and using that $\alpha_p(x) = \alpha_p(-x)$ for all $p \in 2\pi \Z$ 
we see that (\ref{lem:eq:com}) is equivalent to 
\begin{equation}\label{eq:combound}
K_{T}^{0} (-ih\nabla + hp/2) + K_{T}^{0}(-ih\nabla -hp/2) - 2 a \, \delta(x/h) \geq \frac 2 C h^2 p^2
\end{equation}
for all $p\in 2\pi \Z$. This inequality holds for all $p\in \R$, in fact, for
an appropriate choice of $C>0$, as we shall now show.

Since $K_{T}^{0} \geq \const (1+ h^2(-i\nabla + p/2)^2)$, it suffices to
consider the case of $h p$ small. If
$\kappa=\Delta_0/\tanh\left[\frac{\Delta_0}{2T}\right] \geq 2T$ denotes the gap
in the spectrum of $K^{0}_{T}(-ih\nabla) - ah\delta$ above zero, and
$h^{-1/2}\phi_0(x/h)$ its normalized ground state, proportional to
$\alpha^0(x/h)$,
\begin{align}\nn
& K_{T}^{0} (-ih\nabla + hp/2)+ K_{T}^{0} (-ih\nabla - hp/2)-2 a h 
\delta \\ \nn & \geq  \kappa\left[ e^{i hxp/2}\left( 1 -
|\phi_0\rangle\langle\phi_0|\right) e^{-ihxp/2}+ e^{-ih xp/2}\left( 1 -
|\phi_0\rangle\langle\phi_0|\right) e^{ihxp/2} \right] \\ & \geq \kappa\left[ 1
- \left| \int |\phi_0(x)|^2 e^{-ihxp} dx\right|  \right]\,.
\end{align}
In order to see the last inequality, simply rewrite the term as $\kappa(2 - |f\rangle\langle f| -|g\rangle\langle g|)$, where $|\langle f | g \rangle|^2 = \left| \int |\phi_0(x)|^2 e^{-ihxp} dx\right|$, and compute the smallest eigenvalue of the corresponding $2 \times 2$ matrix. Since $\phi_0$ is reflection symmetric, normalized and satisfies $\int x^2 |\phi_0|^2\,dx<\infty$ (see Step 4 in the proof of Proposition \ref{decomp}), we have
$$
1 - \left| \int |\phi_0(x)|^2 e^{-ihxp} dx\right| 
= \int |\phi_0(x)|^2 \left(1-\cos(hpx)\right) dx
\geq c h^2 p^2 \,.
$$
This completes the proof of \eqref{eq:combound}.
\end{proof}


\subsection{The lower bound}

Pick a $\Gamma$ with $\F^{\rm BCS}(\Gamma)\leq \F^{\rm BCS}(\Gamma^0)$ and
let $\psi$ be as in Proposition \ref{decomp} (depending on some parameter
$\epsilon\geq h$ to be chosen later). As before we let $\Delta(x) = -\psi(x)
\Delta_0$ and define $H_\Delta$ by \eqref{hdelta}. We also put $\Gamma_\Delta =
(1+ \exp(\beta H_\Delta))^{-1}$.

Our starting point is the representation
\begin{align}  \label{eq1}
& \F^{\rm BCS}(\Gamma) - \F^{\rm BCS}(\Gamma_0) = -\frac T 2  \Tr\left[ \ln(1+e^{-\beta H_\Delta})-\ln(1+e^{-\beta H_0})\right] \nonumber\\
 & \qquad\qquad + \frac T2 \, \H(\Gamma,\Gamma_\Delta) 
 + \Delta_0 \re \int_{\calC} \overline{\psi(x)} \alpha(x,x) dx - h a
\int_{\calC} |\alpha(x,x)|^2 dx \,.
\end{align}
(Compare with \eqref{eq:repr}.) According to the 
decomposition \eqref{eq:alphadecomp} which, in view of the BCS gap equation
\eqref{BCSgapequ}, reads on the diagonal
$$
\alpha(x,x) = \psi(x) \alpha^0(0) + \sigma(x,x) = \frac{\Delta_0\psi(x)}{2ah}+
\sigma(x,x) \,,
$$ 
we can obtain the lower bound
\begin{align*}
\F^{\rm BCS}(\Gamma) - \F^{\rm BCS}(\Gamma_0) \geq & -\frac T 2  \Tr\left[
\ln(1+e^{-\beta H_\Delta})-\ln(1+e^{-\beta H_0})\right] \nonumber\\
 & + \frac T2 \, \H(\Gamma,\Gamma_\Delta) - ah \int_{\calC}
|\sigma(x,x)|^2\,dx \,.
\end{align*}
For the first two terms on the right side we apply the semiclassics from
Theorem~\ref{sc}. Arguing as in the proof of the upper bound and taking into
account the bounds on $\psi$ from Proposition \ref{decomp} we obtain
\begin{align}\label{eq:lowersc}
\F^{\rm BCS}(\Gamma) - \F^{\rm BCS}(\Gamma_0) 
\geq & h^3 \left(\E^{\rm GL}(\psi) - b_3 \right) - C \epsilon^2 h^{3}
\nonumber \\
& + \tfrac T2 \, \H(\Gamma,\Gamma_\Delta) - ah \int_{\calC} |\sigma(x,x)|^2\,dx
\,.
\end{align}
Our final task is to bound the last two terms from below. In the remainder of
this subsection we shall show that
\begin{equation}\label{eq:remupper}
\tfrac T2 \, \H(\Gamma,\Gamma_\Delta) - ah \int_{\calC} |\sigma(x,x)|^2\,dx
\geq -C \left( \epsilon h^3 + \epsilon^{-2} h^4 \right)  \,.
\end{equation}
The choice $\epsilon=h^{1/3}$ will then lead to
$$
\F^{\rm BCS}(\Gamma) - \F^{\rm BCS}(\Gamma_0) 
\geq h^3 \left(E^{\rm GL} - b_3 \right) - C h^{3+1/3} \,,
$$
which is the claimed lower bound.

In order to prove \eqref{eq:remupper} we again use the lower bound on the
relative entropy from Lemma~\ref{lem:klein} to estimate
\begin{equation}\label{eq:klein0}
 T \, \H(\Gamma,\Gamma_\Delta) \geq  \Tr\left[ \frac{H_\Delta}{\tanh \tfrac 1{2T
} H_\Delta}\left( \Gamma-\Gamma_\Delta\right)^2 \right] \,.
\end{equation}
The next lemma will allow us to replace the operator $H_\Delta$ in this bound
by $H_0$.

\begin{Lemma}
 There is a constant $c>0$ such that for all sufficiently small $h>0$
\begin{equation}\label{eq:hdeltah0}
\frac {H_\Delta}{\tanh \tfrac 1{2 T} H_\Delta} 
\geq (1-ch) \, K^0_T(-ih\nabla)\otimes \id_{\C^2}  \,.
\end{equation}
\end{Lemma}

\begin{proof}
An application of Schwarz's inequality yields that for every $0<\eta<1$
$$
H_\Delta^2 \geq (1-\eta) \left(H_{\Delta_0}^0\right)^2- \eta^{-1} (  \Delta_0^2 \|\psi - 1\|_\infty^2 + h^2 \|W\|_\infty^2) \,.
$$
The expansion formula \cite[(4.3.91)]{AS}
$$
 \frac{ x}{\tanh(x/2)} = 2 + \sum_{k=1}^\infty \left( 2- \frac{2 k^2 \pi^2}{x^2/4 + k^2 \pi^2}\right)
$$
shows that $x \mapsto \sqrt x/ \tanh \sqrt x$ is an operator monotone function. This operator monotonicity implies that 
\begin{align*}
 K^0_T(-ih\nabla)\otimes \id_{\C^2} &\leq \frac { (1-\eta)^{-1/2}
\sqrt{H_\Delta^2 +
\eta^{-1}(\Delta_0^2\|\psi-1\|_\infty^2+h^2\|W\|_\infty^2)}}{\tanh\tfrac 1{2T} 
(1-\eta)^{-1/2} \sqrt{H_\Delta^2 +
\eta^{-1}(\Delta_0^2\|\psi-1\|_\infty^2+h^2\|W\|_\infty^2)}} \\ & \leq
(1-\eta)^{-1/2}  \frac {  \sqrt{H_\Delta^2 +
\eta^{-1}(\Delta_0^2\|\psi-1\|_\infty^2+h^2\|W\|_\infty^2)}}{\tanh\tfrac 1{2T } 
\sqrt{H_\Delta^2 + \eta^{-1}(\Delta_0^2\|\psi-1\|_\infty^2+h^2\|W\|_\infty^2)}}
\\ & \leq (1-\eta)^{-1/2} \left( 1+ \tfrac 1{4T^2 \eta} 
(\Delta_0^2\|\psi-1\|_\infty^2+h^2\|W\|_\infty^2)\right) \frac {H_\Delta}{\tanh
\tfrac 1{2 T} H_\Delta}
\end{align*}
for $0<\eta<1$. The Sobolev inequality and \eqref{eq:psi} show that $\|\psi\|_\infty \leq C \|\psi\|_{H^1} \leq C$, and hence the lemma follows by choosing $\eta=h$.
\end{proof}

To proceed, we denote $\alpha_\Delta=[\Gamma_\Delta]_{12}$ and  recall from Theorem
\ref{scop} that
\begin{align*}
\alpha_\Delta & = \frac{\Delta_0}{4} \left(\psi K_T^0(-ih\nabla)^{-1} +
K_T^0(-ih\nabla)^{-1} \psi \right) + \eta_1+\eta_2 \\
& = \frac{1}{2} \left(\psi \alpha^0 + \alpha^0 \psi \right) + \eta_1+\eta_2
\end{align*}
with $\eta_1$ and $\eta_2$ satisfying the bounds \eqref{eq:eta1sc} and
\eqref{eq:eta2}. The second equality follows from the explicit form
\eqref{eq:alpha0expl} of $\alpha^0$. Comparing this with \eqref{eq:alphadecomp}
we infer that
\begin{equation}\label{eq:alphadeltadecomp}
\alpha=\alpha_\Delta+\sigma-\eta_1-\eta_2 \,.
\end{equation}
Then \eqref{eq:klein0} and \eqref{eq:hdeltah0} imply that
\begin{align}\label{eq:remupper1}
& \tfrac T2 \, \H(\Gamma,\Gamma_\Delta) - ah \int_{\calC} |\sigma(x,x)|^2\,dx
\notag \\
& \qquad \geq (1-ch)\Tr
K^0_T(-ih\nabla)(\alpha-\alpha_\Delta)(\overline{\alpha-\alpha_\Delta}) - ah
\int_{\calC} |\sigma(x,x)|^2\,dx \notag\\
& \qquad \geq (1-ch)\Tr K^0_T(-ih\nabla)\sigma\overline\sigma - ah \int_{\calC}
|\sigma(x,x)|^2\,dx \notag \\
& \qquad\qquad - (1-ch) 2\re \Tr
K^0_T(-ih\nabla)\sigma(\overline{\eta_1+\eta_2}) \,.
\end{align}
In order to bound the first term on the right side from below we are going
to choose a parameter $\rho\geq0$ such that $ch+\rho\leq 1/2$. Here $c$ is the
constant from \eqref{eq:hdeltah0}. (Eventually, we will pick either $\rho=0$ or
$\rho=1/4$, say.) Note that
\begin{align*}
 (1-ch) \Tr K^0_T(-ih\nabla) - ah \delta 
= \rho K^0_T(-ih\nabla) & +
(1-2ch-2\rho)(K^0_T(-ih\nabla)- ah \delta) \\
& + (ch+\rho) (K^0_T(-ih\nabla)- 2ah \delta) \,.
\end{align*}
We recall that the operator $K^0_T(-ih\nabla)- ah \delta$ is non-negative
and that the operator $K^0_T(-ih\nabla)- 2ah \delta$ has a negative eigenvalue
of order one (by the form boundedness of $\delta$ with respect to
$K^0_T(-i\nabla)$). Hence $K^0_T(-ih\nabla)- 2ah \delta \geq -C_1$ with a
constant $C_1$ independent of $h$. (In the following it will be somewhat
important to keep track of various constants, therefore we introduce here a
numbering.) Moreover, using the fact that $K^0_T(-ih\nabla)\geq c_1
(1-h^2\nabla^2)$ we arrive at the lower bound
\begin{align*}
 (1-ch) \Tr K^0_T(-ih\nabla) - ah \delta 
\geq c_1 \rho (1-h^2\nabla^2) - C_1 (ch+\rho) \,,
\end{align*}
which means for the first term on the right side of \eqref{eq:remupper1} that
\begin{equation}
 \label{eq:remupper3}
(1-ch)\Tr K^0_T(-ih\nabla)\sigma\overline\sigma - ah \int_{\calC}
|\sigma(x,x)|^2\,dx
\geq c_1 \rho \|\sigma\|_{H^1}^2 - C_1 (ch+\rho) \|\sigma\|_2^2 \,.
\end{equation}

We now turn to the second term on the right side of \eqref{eq:remupper1}.
Theorem \ref{scop}, together with the bounds \eqref{eq:psi} on $\psi$, implies
that
$\|\eta_1+\eta_2\|_{H^1}^2 \leq C \epsilon^2 h^3$. This bound, combined with
$\|\sigma\|_{H^1}^2 \leq C \epsilon^{-2} h^{3}$ from Lemma \ref{decomp},
however, is not good enough. (It leads to an error of order $h^3$.) Instead, we
shall make use of the observation that in the decompositions
$\sigma=\sigma_1+\sigma_2$ and $\eta_1+\eta_2$ one has
$$
\Tr K^0_T(-ih\nabla)\sigma_2\overline{\eta_1} = 0 \,.
$$
This can be seen by writing out the trace in momentum space and recalling that
the Fourier transform of the $\psi$ involved in $\sigma_2$ has support in
$\{|p|\geq\epsilon h^{-1}\}$, whereas the one of the $\psi$ involved in $\eta_1$
has support in $\{|p|<\epsilon h^{-1}\}$ (see also \eqref{def:eta1}).

Using the estimates \eqref{eq:beta1} and \eqref{eq:eta2} on $\sigma_1$
and $\eta_2$ we conclude that
\begin{align}
 \label{eq:remupper4}
\left| \Tr K^0_T(-ih\nabla)\sigma(\overline{\eta_1+\eta_2}) \right|
& \leq \left| \Tr K^0_T(-ih\nabla)\sigma_1\overline{\eta_1} \right| 
+ \left| \Tr K^0_T(-ih\nabla)\sigma\overline{\eta_2} \right| \notag \\
& \leq C_2 \left( \epsilon h^3 + h^{5/2} \|\sigma\|_{H^1} \right) \,.
\end{align}
Combining \eqref{eq:remupper1}, \eqref{eq:remupper3} and \eqref{eq:remupper4} we
find that
\begin{align}\label{eq:remupper5}
& \tfrac T2 \, \H(\Gamma,\Gamma_\Delta) - ah \int_{\calC} |\sigma(x,x)|^2\,dx
\notag \\
& \qquad \geq c_1 \rho \|\sigma\|_{H^1}^2 - C_1 (ch+\rho) \|\sigma\|_2^2
- 2 C_2 \left( \epsilon h^3 + h^{5/2} \|\sigma\|_{H^1} \right) \,.
\end{align}
Next, we are going to distinguish two cases, according to whether
$4 C_1 \|\sigma\|_2^2 \leq c_1 \|\sigma\|_{H^1}^2$ or not. In the first case, we
choose $\rho=1/4$ and $h$ so small that $ch+\rho\leq 1/2$. In this way we can
bound the previous expression from below by
$$
\tfrac18 c_1 \|\sigma\|_{H^1}^2 - 2 C_2 \left( \epsilon h^3 + h^{5/2}
\|\sigma\|_{H^1} \right)
\geq - 8 c_1^{-1} C_2^2 h^5 - 2 C_2 \epsilon h^3 \,.
$$
This proves the claimed (indeed, a better) bound \eqref{eq:remupper} in
this case.

Now assume, conversely, that $4C_1 \|\sigma\|_2^2 > c_1 \|\sigma\|_{H^1}^2$.
Then we choose $\rho=0$ and bound \eqref{eq:remupper5} from below by
\begin{align*}
& - C_1 ch \|\sigma\|_2^2 - 2 C_2 \left( \epsilon h^3 + h^{5/2} \|\sigma\|_{H^1}
\right) \\
& \quad \geq - C_1 ch \|\sigma\|_2^2 - 2 C_2 \left( \epsilon h^3 + h^{5/2}
\left(4 C_1/c_1\right)^{1/2} \|\sigma\|_2 \right) \,.
\end{align*}
The bound \eqref{eq:beta} on $\|\sigma\|_2$ now leads again to the claimed lower
bound \eqref{eq:remupper}. 

This concludes the proof of the lower bound to the free energy in
Theorem \ref{thm:main}.  Concerning the statement about approximate
minimizers we note that $ \F^{\rm BCS} (\Gamma^0) - \F^{\rm BCS}
(\Gamma_0) = O(h^3) $ and that our a-priori bounds on $\alpha$ in
Proposition \ref{decomp} remain true under the weaker condition that
$\F^{\rm BCS}(\Gamma) \leq \F^{\rm BCS}(\Gamma^0) + C h^3.$ We leave
the details to the reader.


\section{Proof of semiclassical asymptotics}\label{proof-lemmas}

In this section we shall sketch the proofs of Theorems \ref{sc} and
\ref{scop} containing the semiclassical asymptotics. We shall skip
some technical details and refer to \cite{FHSS} for a thorough discussion.


\subsection{Preliminaries}

It will be convenient to use the following abbreviations
\begin{equation}\label{def:k}
k = -h^2\nabla^2 - \mu + h^2 W(x), \qquad k_0= -h^2\nabla^2 - \mu \,.
\end{equation}
We will frequently have to bound various norms of the resolvents
$(z-k)^{-1}$ for $z$ in the contour $\Gamma$ defined by $\im z=\pm
\pi/(2\beta)$ for $\beta >0$. We state these auxiliary bounds
separately.

For $p\geq 1$, we define the $p$-norm of a periodic operator $A$ by
\begin{equation}
\|A\|_p = \left( \Tr |A|^p \right)^{1/p}
\end{equation}
where $\Tr$ stands again for the trace per unit volume. We note that for a Fourier multiplier $A(-ih\nabla)$, these norms are given as
\begin{equation}
\left\| A(-ih\nabla) \right\|_p = h^{-1/p} \left( \int_{\R} |A(q)|^p \frac{ dq}{2\pi} \right)^{1/p}\,.
\end{equation}
The usual operator norm will be denoted by $\|A\|_\infty$. 

\begin{Lemma}
For $z=t\pm i\pi/(2\beta)$ and all sufficiently small $h$ one has
\begin{equation}\label{pb}
 \left\|(z-k)^{-1}\right\|_p \leq C\, h^{-1/p} \times \left\{ \begin{array}{cl} t^{-1/(2p)}  & \text{for $t\gg 1$} \\ |t|^{-1+1/(2p)} & \text{for $t\ll -1$} \end{array}\right. \quad\text{if}\ 1\leq p\leq\infty\,,
\end{equation}
as well as
\begin{equation}\label{infb}
 \left\|(z-k)^{-1}\right\|_\infty \leq C \times \left\{ \begin{array}{cl} 1  & \text{for $t\gg 1$} \\ |t|^{-1} & \text{for $t\ll -1$} \end{array}\right. \,.
\end{equation}
\end{Lemma}

\begin{proof}
The estimates are easily derived with $k_0$ instead of $k$ by evaluating the
corresponding integral. Since the spectra of $k$ and $k_0$ agree up to $O(h^2)$
the same bounds hold for $k$.
\end{proof}


\subsection{Proof of Theorem~\ref{sc}}\label{sec:sc}

The function $f$ in (\ref{deff}) is analytic in the strip $|\im z|< \pi$, and
 we can write
\begin{align*}
f(\beta H_\Delta)- f(\beta H_0)
= \frac 1{2\pi i}\int_\Gamma f(\beta z) \left[ \frac 1{z-H_\Delta} - \frac
1{z-H_0} \right]\, dz \,,
\end{align*}
where $\Gamma$ is the contour $z= r \pm i\tfrac \pi {2\beta}$, $r\in \R$. We
emphasize that this contour representation is not true for the operators
$f(\beta H_\Delta)$ and $f(\beta H_0)$ separately (because of a contribution
from infinity), but only for their difference. 

We claim that
\begin{equation}
 \label{eq:complconj}
\left[ f(\beta H_\Delta) \right]_{11} = \overline{ \left[ f(\beta H_\Delta)
\right]_{22}} - \beta \overline{ \left[ H_\Delta\right]_{22} } \,.
\end{equation}
Recall that $[\, \cdot\,]_{ij}$ denotes the $ij$ element of an
operator-valued $2\times 2$ matrix. To see \eqref{eq:complconj}, we introduce
the unitary matrix
$$
U = \left( \begin{array}{cc}  0 & 1 \\ -1 & 0 \end{array} \right)
$$
and note that
$$
\left[ f(\beta H_\Delta) \right]_{11} = -\left[ U f(\beta H_\Delta) U
\right]_{22} \,.
$$
On the other hand, $UH_\Delta U = -\overline{H_\Delta}$, which implies that
$$
U f(\beta H_\Delta) U = f(-\beta \overline{H_\Delta}) 
= \overline{f(-\beta H_\Delta)} \,.
$$
The claim \eqref{eq:complconj} now follows from the fact that $f(-z)=f(z)-z$.

Subtracting \eqref{eq:complconj} and the corresponding formula for $H_0$ and
noting that $H_\Delta$ and $H_0$ coincide on the diagonal we find that the two
diagonal entries of $f(\beta H_\Delta)- f(\beta H_0)$ are complex conjugates of
each other. Since their trace is real we conclude that
$$
\Tr\left[ f(\beta H_\Delta)- f(\beta H_0) \right]
= \frac 1{\pi i}\int_\Gamma f(\beta z) \,\Tr \left[ \frac 1{z-H_\Delta} -
\frac 1{z-H_0} \right]_{11} \, dz\,.
$$
(For technical details concerning the interchange of the trace and the integral
we refer to \cite{FHSS}.)

The resolvent identity and the fact that 
\begin{equation}\label{defdel}
\delta:=H_\Delta-H_0=-h \mbox{$\left(\begin{array}{cc} 0 & \psi(x) \\ \overline{\psi(x)} & 0 \end{array}\right)$}
\end{equation}
 is off-diagonal (as an operator-valued $2\times 2$ matrix) implies that 
\begin{align*}
& \Tr \left[ \frac 1{z-H_\Delta} - \frac 1{z-H_0} \right]_{11}
= \Tr \left[ \frac 1 {z-H_0} \left( \delta \frac
1{z-H_0}\right)^2\right]_{11} \\
& \qquad\qquad + \Tr\left[ \frac 1 {z-H_0} \left( \delta \frac 1{z-H_0}\right)^4
\right]_{11}
+ \Tr \left[ \frac 1 {z-H_\Delta} \left( \delta \frac
1{z-H_0}\right)^6\right]_{11} \\
& \qquad =: I_1 + I_2 + I_3 \,.
\end{align*}
In the following we shall prove that
\begin{align}\label{eq:i1}
\frac 1{\pi i}\int_\Gamma f(\beta z) I_1 \,dz & = -\frac{h\beta^2}2 \|\psi\|_2^2
\int_\R g_0(\beta (q^2-\mu))\,\frac{dq}{2\pi} \notag \\
& \quad + \frac{h^3\beta^3}{8} \|\psi'\|_2^2 \int_{\R} \left( g_1(\beta(q^2 -\mu)) + 2\beta q^2 g_2(\beta(q^2-\mu))\right) \, \frac{dq}{2\pi} \notag \\
& \quad + \frac {h^3\beta^3} {2} \langle \psi|W|\psi\rangle  \int_{\R} 
g_1(\beta (q^2-\mu))\, \frac{dq}{2\pi}  \notag\\
& \quad + O(h^5) \|\psi\|_{H^2}^2 \,,
\end{align}
\begin{equation}\label{eq:i2}
\frac 1{\pi i}\int_\Gamma f(\beta z) I_2 \,dz =  \frac{h^3 \beta^3}{8}
\|\psi\|_4^4 \int_{\R} \, \frac{g_1(\beta(q^2-\mu))}{q^2-\mu}\, \frac{dq}{2\pi}
+ O( h^{5}) \|\psi\|^3_{H^1} \|\psi\|_{H^2}
\end{equation}
and
\begin{equation}\label{eq:i3}
\frac 1{\pi i}\int_\Gamma f(\beta z) I_3 \,dz = O(h^5) \|\psi\|_{H^1}^6 \,.
\end{equation}
This will clearly prove \eqref{210}. We will treat the three terms
$I_3$, $I_2$ and $I_1$ (in this order) separately.

\bigskip

${\bf I_3}:$ With the notation $k$ introduced in (\ref{def:k}) at the beginning of this section, we have
\begin{align*}
I_3 = &\Tr \left[\frac 1 {z-H_\Delta} \right]_{11} \Delta \frac 1{z+ k}
\Delta^\dagger \frac 1{z-k} \Delta \frac 1{z+k} \Delta^\dagger \frac
1{z-k}\Delta \frac 1{z+ k} \Delta^\dagger \frac 1{z-k} \,.
\end{align*}
Using H\"older's inequality for the trace per unit volume (see \cite{FHSS}) and
the fact that $|z-H_\Delta|\geq \pi/(2\beta)$, we get
$$
|I_3| \leq \frac {2\beta} {\pi} h^6\|\psi\|_\infty^6
\left\|(z-k)^{-1}\right\|_6^3 \left\|(z+k)^{-1}\right\|_6^3 \,.
$$
Together with \eqref{pb}, this yields
$$
|I_3| \leq \frac{C h^5}{1+|z|^3} \|\psi\|_\infty^6 \,.
$$
Here it was important to get a decay faster than $|z|^{-2}$, since we need to
integrate $I_3$ against the function $f$ which behaves linearly at $-\infty$.
Since $\|\psi\|_\infty \leq C \|\psi\|_{H^1}$ by Sobolev inequalities we have
completed the proof of \eqref{eq:i3}.
\bigskip

${\bf I_2}:$ We continue with
\begin{align*}
I_2 & =\Tr \frac 1{z-k} \Delta \frac 1{z+ k} \Delta^\dagger \frac 1{z-k} \Delta
\frac 1{z+k} \Delta^\dagger \frac 1{z-k} \,.
\end{align*}
By the resolvent identity we have
\begin{equation}\label{res}
\frac 1{z-k} = \frac 1{z-k_0} + \frac 1{z-k_0} h^2 W \frac 1{z-k} \,.
\end{equation}
Using H\"older as above, we can bound 
\begin{align}\nonumber
& \left| \Tr \left(\frac 1{z-k}-\frac 1{z-k_0}\right) \Delta \frac 1{z+ k} \Delta^\dagger \frac 1{z-k} \Delta \frac 1{z+ k} \Delta^\dagger \frac 1{z-k} \right| \\
& \quad \leq h^6 \|W\|_\infty \left\|(z-k_0)^{-1} \right\|_\infty
\|\psi\|_\infty^4 \|(z-k)^{-1}\|_3^3 \|(z+ k)^{-1}\|_\infty^2 \,. \label{f1}
\end{align}
By \eqref{pb} and \eqref{infb} this is bounded by $ C h^5
\|\psi\|^4_{H^1(\calC)} (1+|z|^{5/2})^{-1}$. What we effectively have
achieved for this error is, therefore, to replace one factor of $(z-k)^{-1}$ in $I_2$ by a
factor of $(z-k_0)^{-1}$

In exactly the same way we proceed with the remaining factors $(z-k)^{-1}$ and
$(z+ k)^{-1}$ in $I_2$. The only difference is that $k$ might now be replaced
by $k_0$ in the terms we have already treated, but this does not effect the
bounds.

The final result is that $(\pi i)^{-1} \int_\Gamma f(\beta z) \, I_2 \, dz$
equals
\begin{align*}
\frac 1 {\pi i}\int_\Gamma f(\beta z) \Tr \left[ \frac 1{z-k_0} \Delta 
\frac 1{z+k_0} \Delta^\dagger \frac 1{z-k_0} \Delta \frac 1{z+k_0}
\Delta^\dagger \frac 1{z-k_0} \right] \, dz +  O( h^{5})
\|\psi\|^4_{H^1} \,,
\end{align*}
and it remains to compute the asymptotics of the integral.

Let us indicate how to perform the trace per unit volume $\Tr[ \dots ]$. In
terms of integrals the trace can be written as
\begin{multline}
\frac{h^4}{(2\pi)^4}\int_0^1 dx_1 \int_\R dx_2 \int_\R dx_3 \int_\R dx_4
\int_{\R^4} dp_1 dp_2 dp_3 dp_4 \
\overline{\psi(x_1)} \psi(x_2) \overline{\psi(x_3)} \psi(x_4) \\
\times \frac{e^{i p_1(x_1 - x_2)} }{(z - (h^2 p_1^2 - \mu))^2}   \frac{e^{i
p_2(x_2 - x_3)} }{z +(h^2 p_2^2 - \mu)}   \frac{e^{i p_3(x_3
- x_4)} }{z - (h^2 p_3^2 - \mu)}  \frac{e^{i p_4(x_4 - x_1)} }{z +
(h^2 p_4^2 - \mu)} \,.
\end{multline}
Since $\psi$ is periodic with period one we have
$$
\psi(x_j) = \sum_{l_j \in 2\pi \Z} \hat\psi (l_j) e^{i x_jl_j} \,.
$$
We insert this into the above integral and perform the integrals over $x_2, x_3,
x_4$. This leads to $\delta$-distributions such that we can subsequently perform
the integrals over $p_2, p_3, p_4$, as well as the integral over $x_1$. In this
way we obtain
\begin{align*}
 & \frac 1 {\pi i}\int_\Gamma f(\beta z) \Tr \left[ \frac 1{z-k_0} \Delta 
\frac 1{z+k_0} \Delta^\dagger \frac 1{z-k_0} \Delta \frac 1{z+k_0}
\Delta^\dagger \frac 1{z-k_0} \right] \, dz \\
&\quad = h^3 \sum_{p_1,p_2,p_3 \in 2\pi \Z} \widehat \psi(p_1) \widehat
{\psi^*}(p_2) \widehat\psi(p_3) \widehat {\psi^*}(-p_1-p_2-p_3) 
F(hp_1,hp_2,hp_3)
\end{align*}
with
\begin{align*}
F(p_1,p_2,p_3) &=  \frac {\beta^4}{\pi i} 
\int_\Gamma dz\, f(\beta z) \int_\R \,\frac{dq}{2\pi} \frac 1{\left(
z-\beta((q+p_1+p_2+p_3)^2 + \mu) \right)^2} \\
& \quad \times \frac 1{z+\beta((q+p_1+p_2)^2-\mu)} \,
\frac 1{z-\beta((q+p_1)^2-\mu)} \,
\frac 1{z+\beta(q^2-\mu)} \,.
\end{align*}
The leading behavior is given by
$$
F(0,0,0) \sum_{p_1,p_2,p_3 \in 2\pi \Z} \widehat \psi(p_1) \widehat
{\psi^*}(p_2) \widehat\psi(p_3) \widehat {\psi^*}(-p_1-p_2-p_3)
= F(0,0,0) \|\psi\|_4^4 \,.
$$
The integral $F(0,0,0)$ can be calculated explicitly and we obtain
$$
F(0,0,0) =  \frac{ \beta^3}{8}\int_{\R} \, \frac{g_1(\beta(q^2-\mu))}{q^2-\mu}\, \frac{dq}{2\pi}
$$
with $g_1$ from \eqref{defg1}. In order to estimate the remainder we use the fact that \cite{FHSS}
$$
\left| F(p_1,p_2,p_3) - F(0,0,0)\right| \leq \const \left( p_1^2 + p_2^2 + p_3^2\right) \,.
$$
Using Schwarz and H\"older we can bound
$$
\sum_{p_1,p_2,p_3 \in 2\pi \Z} p_1^2 \left| \widehat \psi^*(p_1) \widehat \psi^*(p_2) \widehat\psi(p_3) \widehat \psi(-p_1-p_2-p_3)\right| \leq \const \|\psi\|_{H^2}\|\psi\|_{H^1}^3
$$
and equally with $p_1^2$ replaced by $p_2^2$ and $p_3^2$. Hence we conclude that
\begin{align*}
 \frac 1{\pi i} \int_\Gamma f(\beta z) \, I_2 \, dz & 
= h^3 \!\!\!\!\! \sum_{p_1,p_2,p_3 \in 2\pi \Z} \!\!\!\! \widehat \psi(p_1)
\widehat {\psi^*}(p_2) \widehat\psi(p_3) \widehat {\psi^*}(-p_1-p_2-p_3) 
F(hp_1,hp_2,hp_3) \\
& \qquad\qquad + O( h^{5}) \|\psi\|^4_{H^1} \\
& = h^3 F(0,0,0) \|\psi\|_4^4 + O( h^{5}) \|\psi\|^3_{H^1} \|\psi\|_{H^2} \,.
\end{align*}
This is what we claimed in \eqref{eq:i2}.

\bigskip

${\bf I_1}:$ Finally, we examine the contribution of
$$
I_1 = \Tr \left[ \frac 1{z-k} \Delta \frac 1{z+ k} \Delta^\dagger \frac 1{z-k}
\right] \,.
$$ 
Using the resolvent identity (\ref{res}) we can write $I_1=I_1^a+I_1^b+I_1^c$, 
where
$$
I_1^a=  \Tr \left[\frac 1{z-k_0} \Delta \frac 1{z+k_0} \Delta^\dagger \frac
1{z-k_0} \right]
$$ 
and
\begin{align*}
I_1^b= \Tr \Biggl[& \frac 1{z-k_0} (k-k_0)  \frac 1{z-k_0} \Delta \frac
1{z+k_0} \Delta^\dagger \frac 1{z-k_0} \\ &   + \frac 1{z-k_0} \Delta  \frac
1{z+k_0} (k_0- k) \frac 1{z+k_0} \Delta^\dagger \frac 1{z-k_0} \\ &  + \frac
1{z-k_0} \Delta  \frac 1{z+k_0} \Delta^\dagger \frac 1{z-k_0} (k-k_0) \frac
1{z-k_0} \Biggl] \,.
\end{align*}
The part $I_1^c$ consists of the rest. We claim that
\begin{align}
 \label{eq:i1a}
\frac 1{\pi i} \int_\Gamma f(\beta z) \, I_1^a \, dz 
& = -\frac{h\beta^2}2 \|\psi\|_2^2 \int_\R g_0(\beta (q^2-\mu))\,\frac{dq}{2\pi} \notag \\
& \quad + \frac{h^3\beta^3}{8} \|\psi'\|_2^2 \int_{\R} \left( g_1(\beta(q^2
-\mu)) + 2\beta q^2 g_2(\beta(q^2-\mu))\right) \, \frac{dq}{2\pi} \notag \\
& \quad + O(h^5) \|\psi\|_{H^2}^2 \,,
\end{align}
\begin{equation}
 \label{eq:i1b}
\frac 1{\pi i} \int_\Gamma f(\beta z) \, I_1^b \, dz =
\frac {h^3\beta^3} {2} \langle \psi|W|\psi\rangle  \int_{\R}  g_1(\beta
(q^2-\mu))\, \frac{dq}{2\pi}  + O(h^5) \|\psi\|_{H^2}\|\psi\|_{H^1}
\end{equation}
and
\begin{equation}
 \label{eq:i1c}
\frac 1{\pi i} \int_\Gamma f(\beta z) \, I_1^c \, dz = O(h^5) \|\psi\|_{H^1}^2
\,.
\end{equation}
Clearly, this will imply \eqref{eq:i1}.

We begin with $I_1^c$. These terms contain at least five resolvents, where at
least two terms are of the form $(z-k_\#)^{-1}$ and at least one term of the
form $(z+k_\#)^{-1}$. (Here $k_\#$ stands for any of the operators $k$ or
$k_0$.) Moreover, they contain at least two factors of $k-k_0$. The terms are
either of the type
\begin{equation}
 \label{eq:i1c1}
A=\Tr  \frac 1{z-k_0} (k-k_0)  \frac 1{z-k} \Delta \frac 1{z+k_0}(k_0-k) \frac 1{z+k} \Delta^\dagger \frac 1{z-k_0}
\end{equation}
(at least three minus signs) or of the type
\begin{equation}
 \label{eq:i1c2}
B=\Tr \frac 1{z-k_0} \Delta  \frac 1{z+k_0} (k_0-k) \frac
1{z+k_0} (k_0-k) \frac1{z+k} \Delta^\dagger \frac 1{z-k}
\end{equation}
(only two minus signs). Terms of the first type we bound by 
$$
|A| \leq C h^6 \|W\|_\infty^2\|\psi\|_\infty^2 \|(z-k_0)^{-1}\|_\infty^2
\|(z+k_0)^{-1}\|_3   \|(z+k)^{-1}\|_3  \|(z-k)^{-1}\|_3 \,.
$$
By \eqref{pb} and \eqref{infb} this can be estimated by $Ch^5 |z|^{-2 + 1/6}$
if $\re z\geq 1$ and by $Ch^5 |z|^{-2-1/6}$ if $\re z\leq -1$. This bound is
finite when integrated against $f(\beta z)$.

Terms of type $B$ can be bounded similarly by replacing $z$ by $-z$. Indeed, we
note that since $\int_\Gamma z\, B \, dz = 0$, we can replace $f(\beta z)$ by
$f(-\beta z) = f(\beta z)-\beta z$ in the integrand without changing the value
of the integral. I.e., we can integrate $B$ against a function that decays
exponentially for negative $t$ and increases linearly for positive $t$, instead
of the other way around. These considerations lead to the estimate
\eqref{eq:i1c}.

\medskip

Next, we discuss the term $I_1^a$. After doing the contour integral the term $I_1^a$ gives  
$$
(\pi i)^{-1} \int_\Gamma f(\beta z) \, I_1^a \, dz 
= h \sum_{p \in 2\pi \Z} | \hat \psi(p) |^2  G (hp)
$$
with
$$
G(p) = -\frac{\beta}{2}
\int_{\R} \frac { \tanh\left(\tfrac 1{2} \beta((q+p)^2-\mu)\right) +   \tanh\left(\tfrac 1{2}\beta (q^2-\mu)\right)}{(p+q)^2 + q^2 -2\mu}\, \frac{dq}{2\pi} \,.
$$
By definition (\ref{defg0}) we have 
$$
G(0) = - \frac {\beta^2}{2}  \int_{\R} g_0(\beta (q^2-\mu))\,\frac{dq}{2\pi} \,.
$$
Integrating by parts we can write
\begin{align}\nn 
G''(0) = \frac {\beta^3}{4} \int_{\R} \left( g_1(\beta(q^2 -\mu)) + 2\beta q^2
g_2(\beta(q^2-\mu))\right) \, \frac{dq}{2\pi}
\end{align}
with $g_1$ and $g_2$ from \eqref{defg1} and \eqref{defg2}. Moreover, one can show that \cite{FHSS}
$$
\left|G(p)-G(0)-\tfrac12 p^2 G''(0) \right| \leq C p^4 \,.
$$
{}From this we conclude that
\begin{align*}
\frac 1{\pi i} \int_\Gamma f(\beta z) \, I_1^a \, dz 
& = h \sum_{p \in 2\pi \Z} | \hat \psi(p) |^2  \left( G(0) + \tfrac12 G''(0) h^2 p^2\right) + O(h^5) \|\psi\|_{H^2}^2 \\
& = h G(0) \|\psi\|_2^2 + \tfrac12 G''(0) h^3 \|\psi'\|_2^2 + O(h^5) \|\psi\|_{H^2}^2 \,,
\end{align*}
which is what we claimed in \eqref{eq:i1a}.

\medskip

Finally, we proceed to $I_1^b$. After the contour integration we find
$$
\frac 1{\pi i} \int_\Gamma f(\beta z) \, I_1^b \, dz 
= h^3 \sum_{p,q \in 2\pi \Z} \widehat \psi^*(p) \widehat \psi(q) \widehat
W(-p-q) L(hp,hq) \,,
$$
where
\begin{align*}
L(p,q) & = \beta^3 \int_{\R}  L(p,q,k) \, \frac{dk}{2\pi} 
\end{align*}
with
\begin{align*}
L(p,q,k) & =   \frac 1{\pi i} \int_\Gamma  \ln\left(2+e^{- \beta z}+e^{\beta  z}\right)  \frac 1{z+k^2-\mu} \frac 1{z-p^2+\mu} \frac 1{z-q^2+\mu} \\ &
\qquad\qquad  \times \left(\frac 1{z-p^2+\mu} + \frac 1{z-q^2+\mu}+\frac
1{z+k^2-\mu}\right) \, dz \,.
\end{align*}
We have 
$$
L(0,0) = \frac {\beta^3} {2} \int_{\R}  g_1(\beta (k^2-\mu))\, \frac{dk}{2\pi}
$$
and (see \cite{FHSS} for details)
$$
\left| L(p,q) - L(0,0) \right| \leq C  \left( p^2+q^2\right) \,.
$$
By the Schwarz inequality we can bound
$$
\sum_{p,q \in 2\pi \Z} \left| \widehat \psi^*(p) \widehat \psi(q) \widehat
W(-p-q) (p^2 + q^2) \right|
\leq C \|W\|_2 \|\psi\|_{H^2}\|\psi\|_{H^1} \,,
$$
and obtain
\begin{align*}
\frac 1{\pi i} \int_\Gamma f(\beta z) \, I_1^b \, dz 
& = h^3 L(0,0) \!\sum_{p,q \in 2\pi \Z}\! \widehat \psi^*(-p) \widehat \psi(q)
\widehat W(p-q)
+ O(h^5) \|\psi\|_{H^2}\|\psi\|_{H^1} \\
& = \frac {h^3\beta^3} {2} \langle \psi|W|\psi\rangle  \int_{\R}  g_1(\beta (k^2-\mu))\, \frac{dk}{2\pi}  + O(h^5) \|\psi\|_{H^2}\|\psi\|_{H^1} \,.
\end{align*}
This concludes the proof of Theorem \ref{sc}.


\subsection{Proof of Theorem~\ref{scop}}\label{sec:scop}

Since the function $\rho$ in \eqref{defrho} is analytic in the strip $|\im z| < \pi$, we can write $\left[\rho(\beta H_\Delta)\right]_{12}$ with the aid of a contour integral representation as
\begin{equation}
\left[\rho(\beta H_\Delta)\right]_{12} = \frac 1{2\pi i} \int_\Gamma  \rho(\beta z) \left[ \frac 1{z-H_\Delta} \right]_{12} \,dz\,,
\end{equation}
where $\Gamma$ is again the contour $\im z = \pm \pi/(2\beta)$.
We expand $(z-H_\Delta)^{-1}$ using the resolvent identity and note that, since
$H_\Delta = H_0+\delta$ with a $H_0$  \emph{diagonal} and 
$\delta$  \emph{off-diagonal}, only the terms containing an odd number of $\delta$'s contribute to
the $12$-entry of $(z-H_\Delta)^{-1}$. In this way arrive at the decomposition
\begin{equation}
\left[\rho(\beta H_\Delta)\right]_{12} = \eta_0 + \eta_1 + \eta_2^a +
\eta_2^b \,,
\end{equation}
where
\begin{equation}\label{def:eta0}
\eta_0 = - \frac{h}{4\pi i}\int_\Gamma \rho(\beta z) \left(\psi\ \frac
1{z^2-k_0^2} + \frac{1}{z^2-k_0^2}\ \psi \right) \, dz \,,
\end{equation}
\begin{equation}\label{def:eta1}
\eta_1 =  \frac h{4\pi i}\int_\Gamma \rho(\beta z) \left(\frac
1{z-k_0}\left[\psi, k_0\right] \frac{1}{z^2-k_0^2} + \frac{1}{z^2-k_0^2}
[\psi,k_0]\frac 1{z+k_0} \right) \, dz \,,
\end{equation}
\begin{equation}
\eta_2^a = - \frac {h^3}{2\pi i } \int_\Gamma \rho(\beta z) \frac 1{z-k_0}
\left( W \frac 1{z-k}\psi + \psi\frac 1{z+k_0} W\right) \frac 1{z+ k} \, dz
\end{equation}
and
\begin{equation}
\eta_2^b = - \frac {h^3}{2\pi i}  \int_\Gamma \rho(\beta z) \frac 1{z-k}\psi
\frac 1{z+ k} \bar\psi \frac 1{z-k} \psi  \left[ \frac 1{ z-H_\Delta }
\right]_{22} dz \,.
\end{equation}

A simple residue computation yields
\begin{align*}
\eta_0 & = - \frac{h}{4}  \left( \psi \ \frac{\rho(\beta k_0)-\rho(-\beta
k_0)}{k_0}
+ \frac{\rho(\beta k_0)-\rho(-\beta k_0)}{k_0}\ \psi \right) \\
& = \frac{h\beta}{4}  \left( \psi \ g_0(\beta k_0) + g_0(\beta k_0)\ \psi
\right) \,,
\end{align*}
which is the main term claimed in the theorem. In the following we shall prove that
\begin{equation}\label{eq:eta1}
\|\eta_1 \|_{H^1}^2 \leq C h^5 \|\psi\|_{H^2}^2 \,,
\end{equation}
\begin{equation}\label{eq:eta2a}
\|\eta_2^a \|_{H^1}^2 \leq C h^5 \|\psi\|_{H^1}^2 \,,
\end{equation}
and
\begin{equation}\label{eq:eta2b}
\|\eta_2^b \|_{H^1}^2 \leq C h^5 \|\psi\|_{H^1}^6 \,.
\end{equation}
This clearly implies Theorem \ref{scop}.

\bigskip

$\mathbf{\eta_1}:$ The square of the $H^1$ norm of $\eta_1$ is given by
$$
\|\eta_1\|_{H^1}^2 = h \sum_{p\in 2\pi \Z} |\hat \psi(p)|^2 J(hp)
$$
with
$$
J(p) =  \frac{\beta^4}4 \int_{\R} \left( (q+p)^2 - q^2\right)^2 \left( 1+
q^2\right) \left| F(q+p,q)- F(q,q+p)\right|^2\, \frac{dq}{2\pi}
$$
and 
$$
F(p,q)=\frac 1{p^2-\mu}\frac 1{1+e^{\beta (p^2-\mu)}}\frac
1{1+e^{\beta(q^2-\mu)}} 
\left(\frac{e^{\beta(p^2-\mu)}-e^{\beta(q^2-\mu)}}{p^2-q^2}
+ \frac{e^{\beta(p^2+q^2-2\mu)}-1}{p^2+q^2-2\mu}\right) \,.
$$
One can show that $0\leq J(p)\leq C p^4$ \cite{FHSS}, which yields the desired
bound \eqref{eq:eta1}.

\bigskip

$\mathbf{\eta_2^a}:$ This term is a sum of two terms and we begin by bounding
the first one, that is, $-h^3(2\pi i)^{-1} \int \rho(\beta z)
(z-k_0)^{-1} W (z-k)^{-1}\psi (z+ k)^{-1} \, dz$. Using H\"older's
inequality for the trace per unit volume we find that the square of the $H^1$
norm of the integrand can be bounded by
\begin{align*}
& \Tr\left[ \frac{1-h^2\nabla^2}{|z-k_0|^2} W
\frac{1}{z-k}\psi\frac{1}{|z+k|^2}\overline\psi\frac{1}{\bar z-k} W \right] \\
& \quad \leq \left\| \frac{1-h^2\nabla^2}{|z-k_0|^2} \right\|_\infty
\|W\|_{\infty}^2 \|\psi\|_\infty^2 \|(z-k)^{-1}\|_\infty^2 \|(z+k)^{-1}\|_2^2
\,.
\end{align*}
In order to bound this we use \eqref{pb} and \eqref{infb}, as well as the fact
that $\| (1-h^2\nabla^2) |z-k_0|^{-2}\|_\infty$ is bounded by $C|z|^{-1}$ if
$\re z\leq -1$ and by $C|z|$ if $\re z\geq 1$. (This follows similarly as
\eqref{infb}.)  In particular, we conclude that for $\re z\leq -1$ the previous
quantity is bounded by $C h^{-1} \|\psi\|_\infty^2 |z|^{-7/2}$. The square
root of this is integrable against $\rho(\beta z)$ and we arrive at the bound
$C h^{5/2} \|\psi\|_\infty$ for the $H^1$ norm. For the positive $z$
direction, we notice that $\rho(\beta z)$ decays exponentially leading to a
finite result after $z$ integration. 

For the second term in $\eta_2^a$ we proceed similarly. It is important to first
notice that $\rho(z) = 1-\rho(-z)$, however, and that the $1$ does not
contribute anything but integrates to zero. Proceeding as above we arrive at
\eqref{eq:eta2a}.

\bigskip

$\mathbf{\eta_2^b}:$ Finally, we consider $\eta_2^b$. Using H\"older's
inequality for the trace per unit volume and bounding
$[(z-H_\Delta)^{-1}]_{22}$ by $2\beta /\pi$ for $z\in\Gamma$ we find that the
square of the $H^1$ norm of the integrand is bounded by
\begin{align*}
& \frac{4\beta^2}{\pi^2} \|\psi\|_\infty^6 
\left\| \frac{1-h^2\nabla^2}{|z-k_0|^2} \right\|_\infty
\|(z-k)^{-1}\|_\infty^2 \|(z+k)^{-1}\|_2^2
\,.
\end{align*}
Similarly as in the bound for $\eta_2^a$ one can show that for $\re z\leq -1$
this is bounded by $C h^{-1} \|\psi\|_\infty^6 |z|^{-7/2}$. This leads to
\eqref{eq:eta2b}.

\bigskip

\noindent {\it Acknowledgments.} Part of this work was carried out at
the Erwin Schr\"odinger Institute for Mathematical Physics in Vienna,
Austria, and the authors are grateful for the support and hospitality
during their visit. Financial support via U.S. NSF grants DMS-0800906 (C.H.)
and PHY-0845292 (R.S.) and a grant from the Danish council for
independent research (J.P.S.) is gratefully acknowledged.


\end{document}